%% file: main.tex
\documentclass[12pt]{article}
\usepackage[utf8]{inputenc}

\input{preamble.tex}

\usepackage{xcolor}

\ifdefined\ShowComment
\newcommand{\danupon}[1]{{\bf \color{green} DANUPON: #1}}
\else
\newcommand{\danupon}[1]{}
\fi

\newcommand{\Tres}{T_{rsp}}
\newcommand{\Tcut}{T_{cut}}
\newcommand{\cut}{{\sf cut}}

\title{Weighted Min-Cut: Sequential, Cut-Query and Streaming Algorithms}
\author{Sagnik Mukhopadhyay \thanks{KTH Royal Institute of technology, Sweden, \texttt{sagnik@kth.se}} \and Danupon Nanongkai \thanks{KTH Royal Institute of technology, Sweden, \texttt{danupon@kth.se}}}
\date{}

\begin{document}

\begin{titlepage}
	\maketitle
	\pagenumbering{roman}

\input{abstract}

	\newpage
	\setcounter{tocdepth}{2}
	\tableofcontents
\end{titlepage}

\newpage
\pagenumbering{arabic}

\input{intro}

\subsection{Organization}
First, we provide the preliminaries required in Section \ref{sec:prelim}. In Section \ref{sec:algo-prim}, we provide a schematic algorithm for the minimum 2-respecting cut problem. We first provide the schematic algorithm when the spanning tree is a path in Section \ref{sec:interval}. Later, we extend it the general case in Section \ref{sec:tree-algo} with the help of notions developed in Section \ref{sec:interesting}. We combine this algorithm with Karger's greedy tree packing algorithm to give a weighted min-cut algorithm schematic in Section \ref{sec:karger}. In Section \ref{sec:model-implement}, we provide the implementation of this algorithm in two models: In Section \ref{sec:q-implement}, we provide a graph cut-query implementation (and prove Theorem \ref{thm:intro:query}), and in Section \ref{sec:s-implement} we provide a semi-streaming implementation (and prove Theorem \ref{thm:intro:stream}). Finally, in Section \ref{sec:seq}, we detail the sequential implementation of this algorithm and prove Theorem \ref{thm:intro:seq}.

\section{Preliminaries} \label{sec:prelim}

\paragraph*{Notation.} We denote a spanning tree of a graph $G$ by $T$. We also generally assume that $T$ is rooted, except when $T$ is a simple path. For a rooted tree $T$, we denote the subtree rooted at vertex $u$ as $u^\da$. If the edge $e$ is the parent edge of $u$, \textit{i.e.,} between two vertices of $e$, $u$ is the farthest from the root node, then we sometimes denote the subtree rooted as $u$ as $e^\da$. If two vertices $u$ and $v$ do not belong to the same root-to-leaf path of $T$, we denote them as $u \bot v$. Similarly, if two edges $e_1$ and $e_2$ of $T$ does not belong to the same root-to-leaf path, we denote them as $e_1 \bot e_2$. In a graph $G$, for two disjoint sets of vertices $S$ and $T$, $C_G(S,T)$ denoted the total weight of the edges each of whose one edge point belongs to $S$ and the other endpoint belongs to $T$. By $\cC_G(S,T)$, we denote the set of these edges, \textit{i.e.,} the set of edges going across from $S$ to $T$. When the graph $G$ is clear from the context, we drop the subscript $G$ and denote is as $C(S,T)$ and $\cC(S,T)$ respectively. We use $\deg(S)$ to denote the total weight of the edges whose only one end-point belongs to $S$. To denote the set of such edges, we use the notation $\cC(S)$.

\paragraph*{Reservoir sampling \cite{Vit85}.} We will look at a special randomized sampling technique that will be used in this work, named the \textit{reservoir sampling}. This sampling technique aims to answer the following question:

\begin{question}
Suppose we see a(n infinite) sequence of items $\{a_1, \cdots \}$ and we want to keep $\ell$ many items in the memory with the following invariant: At any point $i$ of the sequence, the $\ell$ items in the memory are sampled uniformly at random from the set $\{a_1, \cdots, a_i\}$. What is a sampling technique that achieves this?
\end{question}

The answer to this question is to use the following sampling method. \begin{itemize}
    \item Store the first $\ell$ items in memory.
    \item From $\ell$-th time period on wards 
    for every $i > \ell$:
        \begin{itemize}
            \item Select the $i$-th item with probability $\ell/i$, and replace a random item from the memory with this new item.
            \item Reject the $i$-th item with probability $(1 - \ell/i)$.
        \end{itemize}
    This makes sure that every item in memory is chosen with probability $\ell/i$.
\end{itemize}

\section{A schematic algorithm for 2-respecting min-cut} \label{sec:algo-prim}
%





In what follows, we first provide a schematic algorithm for finding a 2-respecting weighted min-cut of a graph $G$ which is oblivious to the model of implementation. We then proceed to complete the algorithm for finding min-cut on a weighted graph with an application of cut-sparsification and greedy tree packing. In the subsequent sections, we discuss the complexity of this algorithm when we implement it in different models of computation. Of course, those models need to compute a cut-sparsifier of a given weighted graph efficiently as well---we will show that this is indeed true. We start with a restricted case, where the underlying tree $T$ is a path, and devise an algorithm which handles such graphs. Subsequently we show how we can use this algorithm as a subroutine to achieve an algorithm where there is no assumption of the structure of the underlying spanning tree $T$.

\subsection{When spanning tree is a path} \label{sec:interval}

We look at a slightly different formulation of the problem of finding a 2-respecting weighted min-cut on a graph $G$ where the underlying spanning tree $T$ is a path. We denote it as the \textit{Interval problem}, which is defined below.

\paragraph*{Interval problem.}  Consider $n-1$ points $\{1, \cdots, (n-1)\}$, and order them on a line from left to right. An interval $I = (s,t)$ in $\cI$ (where $s \leq t$) is said to cover a point $i$ if $s \leq v \leq t$ in the path ordering. Given the set of intervals $\cI$, the cost of a pair of point $(i,j)$, denoted as $\cost(i,j)$, is the number of intervals in $\cI$ which covers either $i$ or $j$, but not both. The goal is to find a pair of point $(i,j)$ such that $\cost(i,j)$ is minimized.

\begin{claim} \label{clm:interval-path-eqv}
The Interval problem is equivalent to the problem of finding a 2-respecting min-cut of $G$ where the underlying spanning tree $T$ is a path of length $n$.
\end{claim}

\begin{figure}[h]
\centering
    \includegraphics[scale=0.7]{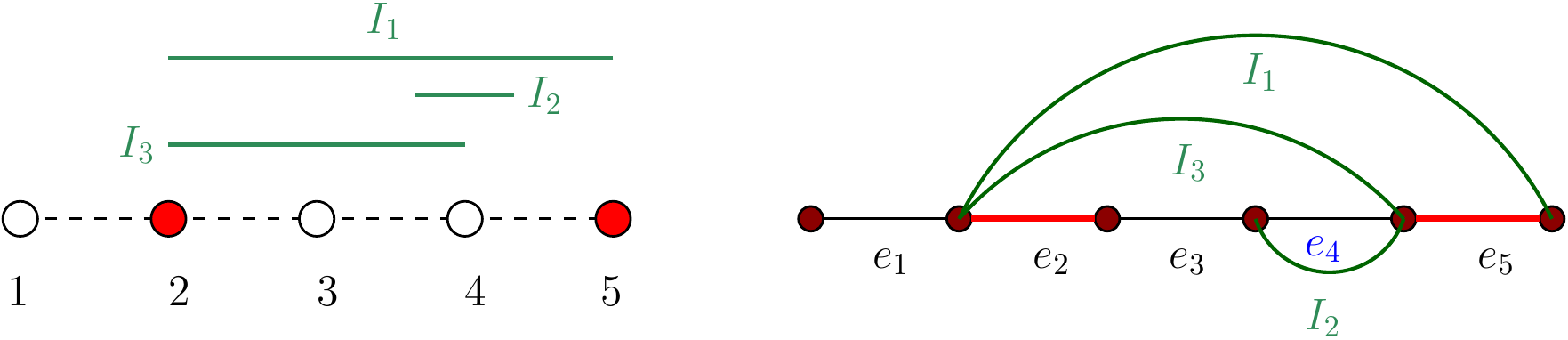}
    \caption{Equivalence between Interval problem and 2-respecting min-cut when the underlying tree is a path of length $n = 6$. There are $5$ points in the Interval problem, and $5$ corresponding edges in the path. The green edges on the RHS corresponds to intervals in the LHS. We are interested in $\cost(2,5)$ (marked in \textcolor{red}{red}).}
        \label{fig:interval}
\end{figure}

\bsni
We defer the proof of this claim to Appendix \ref{app:interval-path-eqv}.

\paragraph*{Cost matrix.} The cost matrix $M_G$ of $G$ with respect to $T$ is defined as a matrix of dimension $(n-1) \times (n-1)$ where the $(i,j)$-th entry of $M_G$ is the weight of a 2-respecting cut of $G$ which respects the $i$-th and the $j$-th edge of $T$. For the Interval problem, $M_G(i,j) = \cost(i,j)$. To reiterate, the goal of the Interval problem is to find the smallest entry in the cost matrix $M_G$.

\begin{figure}[h]
\centering
    \includegraphics[scale=0.7]{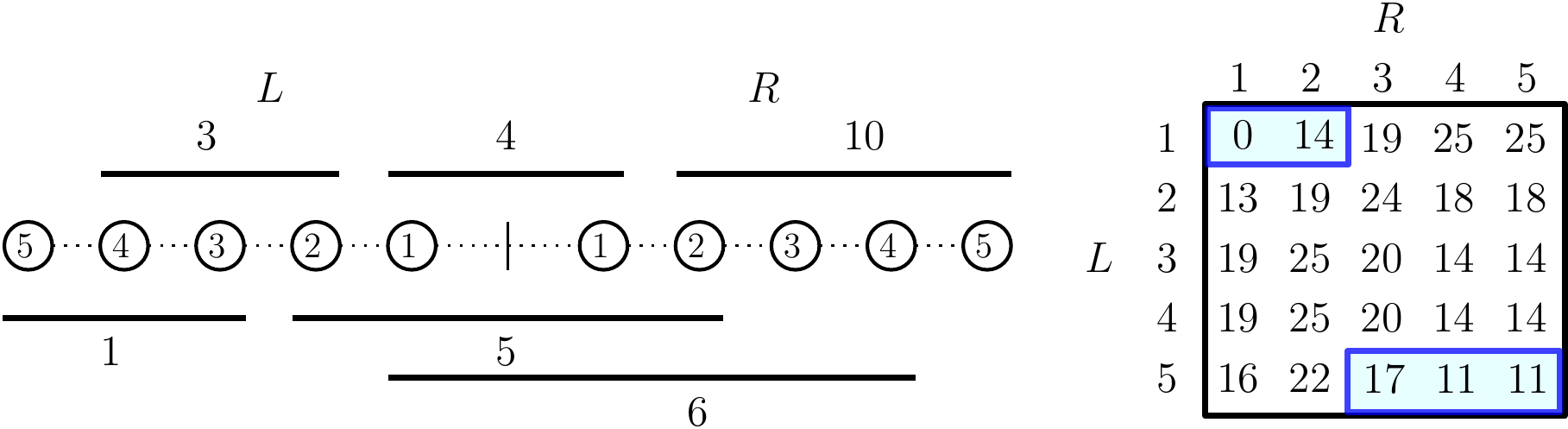}
    \caption{Bipartite Interval problem and corresponding cost-matrix. The \textcolor{blue}{blue boxes} show the positions of the minimum entries in each column.}
        \label{fig:bipart-cost-matrix}
\end{figure}

To this end, we formulate a restricted version of the Interval problem, denoted as \textit{Bipartite interval problem}, as follows: 

\paragraph*{Bipartite Interval problem.} Divide $[n-1]$ in two equal parts $L = \{1, \cdots, \ceil{\frac{n-1} 2}\}$ and $R = \{\ceil{\frac{n-1} 2}+1, \cdots, n-1\}$. We are interested in finding the pair $(i,j) \in L \times R$ such that $\cost(u,v)$ is minimum among such pairs. For solving this problem, we provide a schematic algorithm as follows. The model specific implementations are discussed in subsequent sections. In the algorithm, we assume a slightly unusual ordering on $L$ for simplicity of exposition: Restrict the cost matrix $M$ to $L \times R$ (Denoted by $M_{L \times R}$). The first row corresponds to $\ceil{\frac{n-1} 2}$-th point in $L$, the second row is $(\ceil{\frac{n-1} 2} -1)$-th point, and so on till the $\ceil{\frac{n-1} 2}\}$-th row which is the first point of $L$. The columns are ordered normally in an increasing order from the set $\{\ceil{\frac{n-1} 2}+1, \cdots, n-1\}$. The algorithm is recursive: it is instructive to view the execution of the algorithm as a binary tree of depth $O(\log n)$ where (i) each nodes denotes a recursive call, and (ii) a parent node makes two calls corresponding to two of its children on disjoint sub-problems (or submatrices). In any node $v$, the algorithm reads the middle column ({\sf mid}) of the associated sub-matrix $M_{L_v \times R_v}$. The minimum value of this column can occur at many rows in this column. Let us denote by $i_s^v$ to be the row with the smallest row index where the minimum occurs, and $i_t^v$ to be the row with the largest row index where the minimum occurs. The node $v$ then issues two recursive calls: The left child $u$ is issued with the submatrix $M_{L_u \times V_u}$ and the right child $u'$ is issued with the submatrix $M_{L_{u'} \times V_{u'}}$. Here $L_{u}$ is a prefix of $L_v$ of size $i_s^v$ (i.e., the initial rows of $L_v$ including the row $i_s^v$) and $L_{u'}$ is the suffix of $L_v$ of size $|L_v| - i_t^v + 1$ (i.e., the last few rows of $L_v$ including the row $i_t^v$. Also, $R_{u}$ is a prefix of $R_v$ of size ${\sf mid} -1$ (i.e., the initial columns of $R_v$ not including the column $\sf mid$) and $R_{u'}$ is the suffix of $R_v$ of size ${\sf mid} -1$ (i.e., the last few columns of $R_v$ not including the column $\sf mid$. At the leaf nodes, when either $L_v$ (or $R_v$) is a singleton set, the algorithm read the whole row (or column), and outputs the minimum. Note that, in each depth of the recursion tree, the number of entries of the matrix $M_{L \times R}$ is read is at most $n+1$---this will be important in the analysis of the complexity of the algorithm.

\begin{center}
  \centering
  \begin{minipage}[H]{0.8\textwidth}
\begin{algorithm}[H]
\caption{Schematic algorithm for Bipartite interval problem}\label{alg:bipartite-int}
\begin{algorithmic}[1]
\Procedure{Bipartite-interval}{$L,R$}
\If{$|L| > 1$ and $|R| > 1$}
    \State Let ${\sf mid}$ be the middle column of $R$. 
    \State \textbf{Find} the $i_s, i_t \in L$ such that $\cost(i_s, {\sf mid})= \cost(i_t, {\sf mid})$ is minimum in column ${\sf mid}$, and $i_s$ is the first and $i_t$ is the last such row. \label{algline:min-find} \Statex\Comment{\textcolor{blue}{Read the {\sf mid} column of $M_{L \times R}$}}
    \State Let $L^1$ be the prefix of $L$ of length $i_s$, and $L^2$ be the suffix of $L$ of length $i_t -1$. 
    \State Let $R^1$ be the prefix of $R$ of length ${\sf mid} -1$, and $R^2$ be the suffix of $R$ of length ${\sf mid} - 1$.
    \State \textbf{Run} {\sc Bipartite-interval}$(L^1, R^1)$ and {\sc Bipartite-interval}$(L^2, R^2)$.
\Else
    \State \textbf{Read} the entire $M_{L \times R}$ and \textbf{record} the minimum. \Statex\Comment{\textcolor{blue}{Requires reading at most 1 column or 1 row of matrix $M$}}
\EndIf
\State \textbf{Read} all recorded minimums and \textbf{output} the smallest.
\EndProcedure
\end{algorithmic}
\end{algorithm}
\end{minipage}
\end{center}

Before providing correctness of Algorithm \ref{alg:bipartite-int}, we show how, given Algorithm \ref{alg:bipartite-int}, we can solve the Interval problem.

 \begin{center}
   \centering
  \begin{minipage}[H]{0.8\textwidth}
\begin{algorithm}[H]
\caption{Schematic algorithm for Interval problem}\label{alg:int}
\begin{algorithmic}[1]
\State Set $L = \{1, \cdots, \ceil{\frac{n-1} 2}\}, R = \{\ceil{\frac{n-1} 2}, \cdots, n-1\}$.
\Procedure{Interval}{$L,R$}
\If{$|L| > 1$ and $|R| > 1$}
    \State \textbf{Run} {\sc Bipartite-interval}$(L, R)$ and record the output.
    \Statex \Comment{\textcolor{blue}{Requires reading matrix $M$}}
    \State \textbf{Divide} $L = L^1 \cup L^2$ in two equal parts where $L^1$ is the prefix of $L$ and $L^2$ is the suffix of $L$, both of length $|L|/2$. Similarly, \textbf{divide} $R = R^1 \cup R^2$.
    \State \textbf{Run} {\sc Interval}$(L^1, R^1)$ and {\sc Interval}$(L^2, R^2)$.
\Else
    \State \textbf{Read} the entire $M_{L \times R}$ and \textbf{record} the minimum. \Statex\Comment{\textcolor{blue}{Requires reading constant many entries of $M$}} 
\EndIf
\State \textbf{Read} all recorded outputs and output the smallest.
\EndProcedure
\end{algorithmic}
\end{algorithm}
\end{minipage}
\end{center}

\paragraph*{Correctness of Algorithm \ref{alg:int}.} This follows follows from the observation that, for any $(i,j) \in [n-1] \times [n-1]$, there is a call to the {\sc Interval} procedure where $i \in L$ and $j \in R$. At the $\ell$-th level of the recursion tree, the division of $L = L^1 \cup L^2$ and $R = R^1 \cup R^2$ happens depending on the $\ell$-th most significant bit in the binary representation of $i$ and $j$. This means that $i$ and $j$ will be separated at the $\ell$-th level of recursion (\textit{i.e.,} depth $\ell$ in the associated recursion tree) if the most significant bit where $i$ and $j$ differs is the $\ell$-th bit.

\medskip\noindent
At this point, we make a general claim about the complexity of Algorithm \ref{alg:int}. This claim will be used later to compute the complexity of Algorithm \ref{alg:int} in different models of computation.

\begin{claim} \label{clm:interval-gen-comp}
Let the complexity of finding $i_s$ and $i_t$ in Line \ref{algline:min-find} of Algorithm \ref{alg:bipartite-int} on $n$ points be $\tO(n)$ in some measure of complexity. Then the complexity of Procedure {\sc Interval} is $\tO(n)$.
\end{claim}

\begin{proof}
Let us denote the complexity of Procedure {\sc Bipartite-interval} on input $L$ and $R$ where $|L| = \ell, |R| = r$ be $ T(\ell, r)$. Then, we can write the following recursion:
\[
T(n,n) = \tO(n) + T(n_0, n/2) + T(n_1, n/2),
\]
where $n_0 + n_1 \leq n + 1$ and $T(1, r) = O(r)$ and $ T(\ell,1) = O(\ell)$. At the beginning, $n_0 = i_s$ and $n_1 = n - i_t +1$. Solving this recursion, we get $ T(n,n) = \tO(n)$.

Now let us denote the complexity of {\sc Interval} on $n$ points by $C(n)$. Then we have
\[
C(n) = T(n/2,n/2) + 2 C(n/2),
\]
with $C(1)= O(1)$. Solving this recursion and putting the value of $T(n/2,n/2)$, we get $C(n) = \tO(n)$.
\end{proof}




\subsubsection{Correctness of Algorithm \ref{alg:bipartite-int}}
\label{sec:bipart-interval}

First, we note an important property of the cost-matrix $M$ of the Interval problemon $n-1$ points. Given $M$, we define $n-2$ vectors $\Delta_1, \cdots, \Delta_{n-2}$, each of dimension $n-1$ as follows: $\Delta_j(i) = M(i,j) - M(i, j+1)$.

\begin{claim} \label{clm:delta-increasing}
For each $j$, the vector $\Delta_j$ is monotonically increasing.
\end{claim}

Given Claim \ref{clm:delta-increasing}, we can observe the following property of the cost matrix.\footnote{In literature, such a matrix is referred to as a \textit{Monge matrix} \cite{BurkardKR96}. A few algorithms are known for finding column minima of a Monge matrix, most notably the SMAWK algorithm \cite{AggarwalKMSW87}. We use a divide and conquer approach because it gives better handle for adapting the schematic algorithm in different computational models.} Suppose we are interested in finding the minimum entry in each column. 

\begin{claim} \label{clm:delta-min}
Consider column $j$ and let $i_s$ be the first row and $i_t ( \geq i_s)$ be the last row where minimum occurs in the column $j$. Then, for any column $j' < j$, a minimum entry of column $j'$ will occur at rows in the set $\{1, \cdots, i_s\}$, and for any column $j'' > j$, a minimum of column $j''$ will occur at rows in the set $\{i_t, \cdots, n\}$.
\end{claim}

Given Claim \ref{clm:delta-min}, it is clear that Algorithm \ref{alg:bipartite-int} records minimum entry of each column of the cost matrix. In the end, the algorithm outputs the smallest among these minimum entries, which is the minimum entry of the whole cost-matrix. Next we prove Claim \ref{clm:delta-increasing} and \ref{clm:delta-min}.

\begin{proof}[Proof of Claim \ref{clm:delta-increasing}]
We denote any interval which goes from $L$ to $R$ as \textit{crossing interval}. We also denote any interval contained in $L$ as \textit{$L$-interval} and any interval contained in $R$ as \textit{$R$-interval}. We start with an empty cost-matrix (every entry is 0) and will introduce each interval one by one. Fix any $j$. To start with, the $\Delta_j$ is an all-0 vector and hence monotonically increasing. We show that $\Delta_j$ will maintain this property when we introduce any of the three kinds of intervals. Figure \ref{fig:bipart-interval} shows the contribution of different types of intervals in the cost-matrix.

\begin{figure}[h]
\centering
    \includegraphics[scale=0.7]{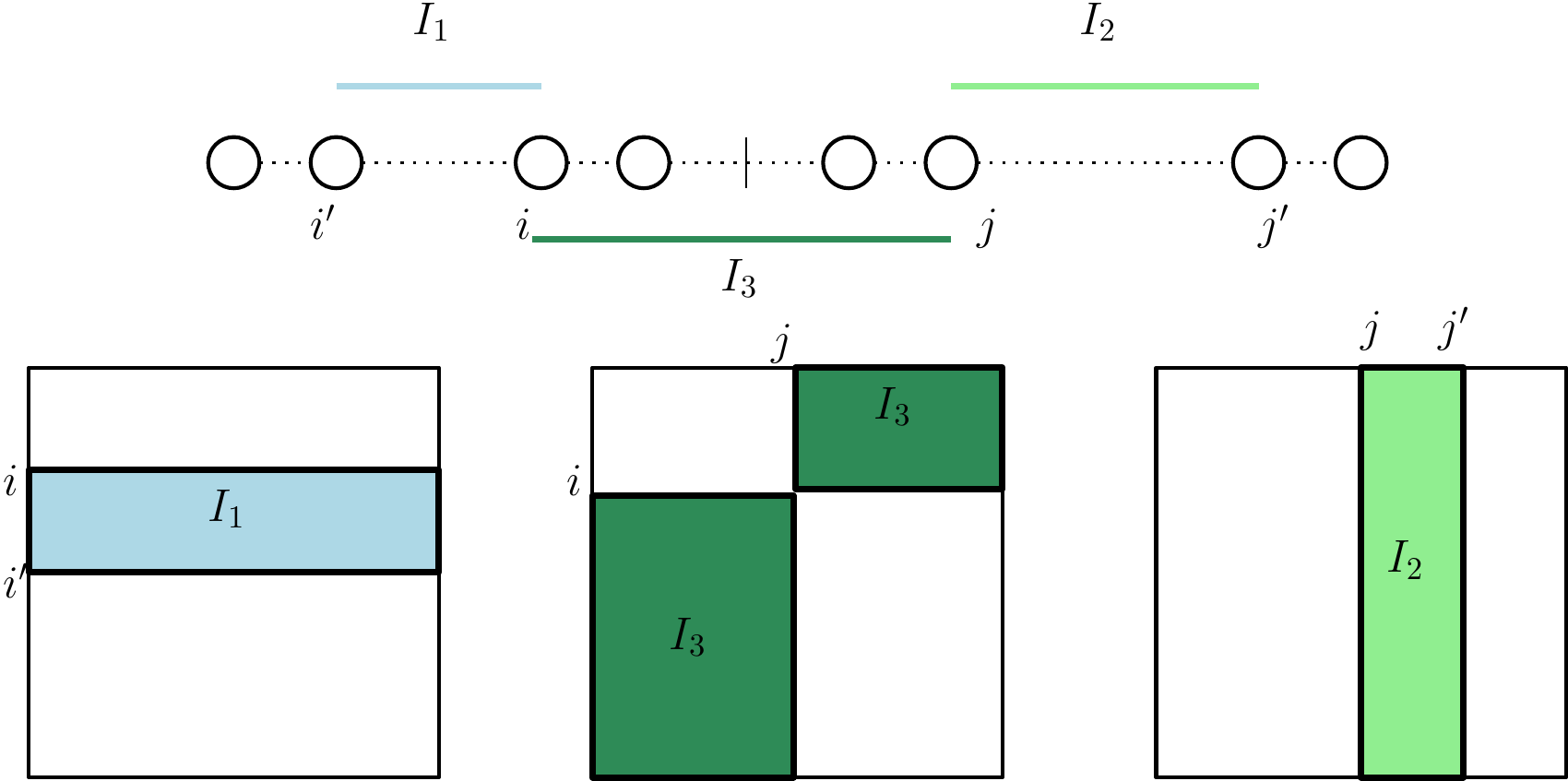}
    \caption{Contribution of each type of Interval ($L, R$, Crossing) in the cost matrix $M$.}
    \label{fig:bipart-interval}
\end{figure}

\paragraph*{$L$-interval.} This type of interval does not increase any entries of $\Delta_j$ as they increase both $M(i,j)$ and $M(i,j+1)$ by the same amount.

\paragraph*{$R$-interval.} If neither of $j$ and $j+1$ is covered by the interval, or both of them are covered by the interval, then $\Delta_j$ does not change. Else, assume that $j$ is covered by the interval. This means $M(i,j)$ is increased by the same value for all $i$. So, if $\Delta_j$ was monotonically increasing before introducing this interval, then $\Delta_j$ remains monotonically increasing. The case where $j+1$ is covered by the interval can be handled similarly.

\paragraph*{Crossing interval.} Let us assume that the interval starts from $p \in L$ and end on $q \in R$. If $q \neq j+1$, then this can be handled similar to $L$-interval on $R$-interval. When $q = j+1$, then $M(i,j+1)$ increases for all $i \leq p$, $M(i,j)$ increases for all $i \geq p$, and  all these increments are by the same amount (i.e., by the weight of the interval). This operation also does not violate the necessary property of $\Delta_j$.
\end{proof}

\begin{proof}[Proof of Claim \ref{clm:delta-min}]
First we show that, for any column $j' < j$, a minimum entry of column $j'$ will occur at rows up to $i_s$. The other case will follow by a similar argument. Let us denote $\Delta_{j',j}$ to be a vector where $\Delta_{j',j} = M(i,j') - M(i, j)$. It is not hard to see that $\Delta_{j',j}(i) = \Delta_{j'} + \cdots + \Delta_{j-1}$, and hence is also monotonically increasing (Claim \ref{clm:delta-increasing}). 

Now, for the sake of contradiction, assume that there is no minimum entry in the column $j'$ in any row in the set $\{1, \cdots, i_s\}$. Let the minimum entry at column $j'$ occurs at row $i > i_s$, \textit{i.e.}, $M(i,j') < M(i_s,j')$. We have $M(i, j') = M(i,j) + \Delta_{j',j}(i) \geq M(i_s,j) + \Delta_{j',j}(i) \geq M(i_s,j) + \Delta_{j',j}(i_s)$. Here the first inequality follows from the fact that, for column $j$, $M(i_s,j)$ is a minimum entry; and the second inequality follows from the fact that $\Delta_{j',j}$ is monotonically increasing. This gives us the necessary contradiction.
\end{proof}

\subsection{\Interest edges and paths} \label{sec:interesting}

This section deals with the notion of \textit{interesting} edges which we will need later to design an algorithm for the general case. From this point onward, we will make connections between edges of a spanning tree $T$ and vertices of $T$. We will follow the labeling below unless stated otherwise: For vertices $u, v, v'$, we denote the edges of $T$ which are parents of these vertices as $e, e', e''$ respectively.

\begin{definition}[Cross-\interest edge] \label{def:cross-int-vertex}
Given an edge $e \in T$, we denote an edge $e' \in T$ ($e \bot e'$) to be cross-\interest with respect to $e$ if $C(u^\da, v^\da) > \frac{\deg(u^\da)} 2$ where $e$ is the parent edge of $u$ and $e'$ is the parent edge of $v$.
\end{definition}

\begin{definition}[Down-\interest edge] \label{def:down-int-vertex}
Given an edge $e$, we denote an edge $e' \in u^\da$ to be down-\interest with respect to $e$ if $C(v^\da, V - u^\da) > \frac{\deg(u^\da)} 2$ where $e$ is the parent edge of $u$ and $e'$ is the parent edge of $v$.
\end{definition}


If two edges $e$ and $e'$ are cross-\interest to each other, or if $e$ is down-\interest to $e'$, then we denote the pair $(e,e')$ as \textit{candidate exact-2-respecting cut}. The reason for denoting so will be clear in the analysis of correctness of Algorithm \ref{alg:tree}. Next we make the following observation about cross and down-\interest edges which we justify subsequently.
\begin{observation} \label{obs:int-vbotv}
Given an edge $e$, there cannot be two edges $e'$ and $e''$ such that $e' \bot e''$ and both $e'$ and $e''$ are cross-\interest or down-\interest with respect to $e$.
\end{observation}

\noindent
To recall, $e$ is the parent edge of $u$, $e'$ is the parent edge of $v$ and $e''$ is the parent edge of $v'$.
\begin{itemize}
    \item \textbf{Cross-\interest.} $C(u^\da, v^\da)$ and $C(u^\da, v'^\da)$ are disjoint for $e' \bot e''$, and hence $\deg(u^\da) \leq C(u^\da, v^\da) + C(u^\da, v'^\da)$. If both $v$ and $v'$ are \interest with respect to $u$, then $C(u^\da, v^\da) + C(u^\da, v'^\da) > \deg(u^\da)$ which is a contradiction. 
    \item \textbf{Down-\interest.} $C(v^\da, V - u^\da)$ and $C(v'^\da, V - u^\da)$ are disjoint for $e' \bot e''$ as well, and hence a similar argument goes through.
\end{itemize}

This means that the set of edges which are cross-\interest w.r.t. $e$ belongs to a root-to-leaf path, and the set of edges which are down-\interest w.r.t. $e$ belongs to a $u$-to-leaf path. 

Note that if $v''$ is the parent of $v$, $C(u^\da, v''^\da) \geq C(u^\da, v^\da)$ and $C(v''^\da, V - u^\da) \geq C(v^\da, V - u^\da)$. So we can make the following observation:

\begin{observation}\label{obs:int-ancestor}
Given an edge $e$, if an edge $e'$ is cross-\interest w.r.t $e$, then all ancestors of $e'$ are \interest w.r.t $e$. Similarly, if an edge $e''$ is down-\interest w.r.t. $e$, then all ancestors of $e''$ (up to but not including $e$) are down-\interest w.r.t $e$.
\end{observation}

Because of Observation \ref{obs:int-ancestor}, we extend Definition \ref{def:down-int-vertex} and \ref{def:cross-int-vertex} to the following:

\begin{definition}[\Interest path] \label{def:interest-path}
Given an edge $e$, we denote a path $p \in T$ to be cross-\interest (or down-\interest) with respect to $e$ if there is an edge $e' \in p$ such that $e'$ is cross-\interest (or down-\interest) with respect to $e$. Given two paths $p_1$ and $p_2$, we denote $p_1$ to be cross-\interest (or down-\interest) w.r.t. $p_2$ is there is an edge $e$ in $p_2$ such that $p_1$ is cross-\interest (or down-\interest) w.r.t. $e$.
\end{definition}

\subsection{Handling general spanning tree} \label{sec:tree-algo}

We use the heavy-light decomposition from \cite{ST83}. Given any rooted tree $T$, the heavy-light decomposition splits $T$ into a set $\cP$ of edge-disjoint paths such that any root-to-leaf path in $T$ can be expressed as a concatenation of at most
$\log n$ sub-paths of paths in $\cP$.

\begin{theorem}[\cite{ST83}] \label{thm:tree-decomp}
For any vertex $v$, the number of paths in $\cP$ that starts from $v$ or any ancestor of $v$ is at most $\log n$. In other words, the number of paths from $\cP$ which edge-intersects any root-to-leaf path is at most $\log n$.
\end{theorem}

Next, we describe the algorithm for finding a 2-respecting min-cut; the pseudo-code is provided subsequently. We assume that, in all models of computations that we are interested in, the spanning tree $T$ is known (\textit{i.e.,} stored in local memory). 

The na\"ive algorithm will go over all possible 2-respecting cuts to find out the smallest among them---we want to minimize such exhaustive search. Note that, we assume, it is rather efficient to check the minimum among 1-respecting cuts. But the number of 2-respecting cuts is $\Omega(n^2)$ and hence we cannot afford to go over all possible 2-respecting cuts to find the minimum if we want to be efficient. To this end, we will examine only those 2-respecting cuts which \textit{have the potential to be the smallest}. We describe the process of finding such potential 2-respecting cuts next. The algorithm consists of 5 steps.
\begin{description}
\item[Step 1. Finding 1-respecting min-cut (Line \ref{algline:1-resp-s} to \ref{algline:1-resp-e}):]The algorithm iterates over every edge $e$ of the spanning tree $T$, and considers the cut which respects $e$. The algorithm records the value of the smallest such cut.

\item[Step 2. Heavy-light path decomposition (Line \ref{algline:heavy-light}):] The algorithm uses heavy-light decomposition on $T$ (\textit{viz.} Theorem \ref{thm:tree-decomp}) to obtain a set of edge-disjoint paths $\cP$. Each root-to-leaf path of $T$ can be split into at most $\log n$ many subpaths from $\cP$.

\item[Step 3 (Line \ref{algline:2-resp-path-s} to \ref{algline:2-resp-path-e}): ] In this step, we consider 2-respecting cuts which respects 2 edges of a path $p \in \cP$. For each path $p$, we \textit{collapse} all other edges of $T$ to the vertices of $p$: The \textit{collapse} operation contracts all other edges of $T$ to form super-vertices, and all edges of $G$ which are incident on any vertex which takes part in a super-vertex will now be considered incident on the super-vertex. This operation does not change the 2-respecting cut value that  we are interested in for the following trivial reason: A 2-respecting cut which respects the edges of the path $p$ will not cut any other edge of $T$. Next we run Algorithm \ref{alg:int} on this collapsed (\textit{residual}) graph to find out the minimum cut which 2-respects $p$.

\item[Step 4. Finding \textit{\interest} edge-pairs (Line \ref{algline:find-interest-s} to \ref{algline:find-interest-e}):] In this step, we find \textit{potential} edge-pairs---a pair of edges which, if respected by a cut, will yield a smaller cut value than the cuts which respects only one edge from the pair.  In Step 3, the algorithm may already have taken care a few of them, especially those pairs of edges which fall on the same path in $\cP$. We are now interested in pairs which fall on different paths of $\cP$. We have already introduced necessary notations: for a candidate pair of edges $(e,e')$, we call $e$ is \interest w.r.t to $e'$ and vice versa. To enumerate the set of candidate pairs, the algorithm iterates over each edge $e$ of $T$, and finds a set of other edges which can form a candidate pair with $e$ (\interest w.r.t $e$). By Observation \ref{obs:int-ancestor}, all such edges fall on a single root-to-leaf path. We check only the top (closest to the root) edges of $\cP$ (also because of Observation \ref{obs:int-ancestor}---if there is an edge in a path which is interesting w.r.t $e$, then the edge in the path which is closest to the root is also interesting w.r.t. $e$); and if any of those edges is \interest w.r.t $e$, we declare the path to be \interest.\footnote{In all models of implementation in this work, for every edge $e$, we will have a few root-to-leaf paths from the previous steps with the guarantee that the actual root-to-leaf which contains all \interest edges w.r.t. $e$ is one of those. At this step, we will verify these paths to figure out which one (if any) is the \interest one. See Lemma \ref{lem:q-pairing} for example.} By Theorem \ref{thm:tree-decomp}, there can be at most $\log n$ many paths in $\cP$ which are \interest w.r.t $e$. We also label $e$ with this interesting set of paths. At the end of this step, we have the following: Consider any pair of paths, $p_1$ and $p_2$, in $\cP$. For each edge $e \in p_1$, for which there is an \interest edge in $p_2$, $e$ is labelled with $p_2$. Similarly, for each edge $e' \in p_2$, for which there is an \interest edge in $p_1$, $e'$ is labelled with $p_1$.

\item[Step 5. Pairing (Line \ref{algline:pairing-s} to \ref{algline:pairing-e}):] In this step, we pair up paths from $\cP$ and look at exact-2-respecting candidate cuts which respects one edge form each path of the pair. Consider a pair $(p_1, p_2)$: If there is no edge in $p_1$ which is marked by $p_2$, or vice versa, or both, we discard this pair. Otherwise, we mark the edges in $p_1$ which are interested in $p_2$, and mark the edges in $p_2$ which are interested in $p_1$ as well. We \textit{collapse} (as in Step 3) every other edge of $T$ so that we are left with a residual graph with edge going across only the marked edges of $p_1$ and $p_2$---this does not change the cut value as we are interested in 2-respecting cuts which implies that no other edge of $T$ takes part in the cut. We run Algorithm \ref{alg:int} on the residual graph and record the smallest 2-respecting cut which respects one edge in $p_1$ and another edge in $p_2$. 
\end{description}
At the end, we compare the recorded cuts from Step 1, 3 and 5, and output the minimum among them.


\begin{center}
  \centering
  \begin{minipage}[H]{0.8\textwidth}
\begin{algorithm}[H]
\caption{Schematic algorithm for 2-respecting min-cut}\label{alg:tree}
\begin{algorithmic}[1]
\For{every edge $e \in T$} \label{algline:1-resp-s} \Comment{\textcolor{blue}{Finding 1-respecting min-cut}}
    \State \textbf{Find} out the value of the cut which respects only $e$ and record it.
\EndFor \label{algline:1-resp-e}
\State Use heavy-light decomposition on $T$ to obtain disjoint set of paths $\cP$.  \label{algline:heavy-light}
\Statex \Comment{{\color{blue}Theorem \ref{thm:tree-decomp}}}
\For{every $p \in \cP$} \label{algline:2-resp-path-s} \Comment{\parbox[t]{.4\linewidth}{\textcolor{blue}{Finding cuts respecting 2 edges in a single $p \in \cP$}}}
    \State \textbf{Collapse} all edges in $T$ which are not in $p$ and \textbf{run} Algorithm \ref{alg:int} on the residual graph. \Comment{\parbox[t]{.4\linewidth}{\textcolor{blue}{Algorithm \ref{alg:int} finds minimum exact-2-respecting cut when $T$ is a path}}}
    \State \textbf{Record} the outputs. 
\EndFor \label{algline:2-resp-path-e} 
\For{each edge $e \in T$} \label{algline:find-interest-s} \Comment{\parbox[t]{.4\linewidth}{\textcolor{blue}{Finding the set of interesting paths for each edge in $T$, See Definition \ref{def:interest-path}}}}
    \State \textbf{Find} out the set $\cP^{cross}_e \subseteq \cP$ which are cross-\interest w.r.t. $e$.
    \State \textbf{Find} out the set $\cP^{down}_e \subseteq \cP$ which are down-\interest w.r.t. $e$.
\EndFor\label{algline:find-interest-e}
\For{every distinct pair $(p_1, p_2) \in \cP \times \cP$} \label{algline:pairing-s}
\Comment{\parbox[t]{.4\linewidth}{\textcolor{blue}{Finding cuts which respects edges in different paths in $\cP$ }}}
    \If{$p_1$ and $p_2$ are cross-interested in each other}
            \State In $p_1$, \textbf{mark} edges which are cross-interested in $p_2$.
            \State In $p_2$, \textbf{mark} edges which are cross-interested in $p_1$.
            \State \textbf{Collapse} unmarked edges of $T$ and \textbf{run} Algorithm \ref{alg:int} on the residual graph. \Statex \Comment{\parbox[t]{.4\linewidth}{\textcolor{blue}{Algorithm \ref{alg:int} finds minimum exact-2-respecting cut when $T$ is a path}}}
            \State \textbf{Record} the output.
    \EndIf
    \If{$p_1$ is down-interested in $p_2$}
        \State In $p_1$, \textbf{mark} edges which are down-interested in $p_2$.
        \State In $p_2$, \textbf{mark} all edges.
        \State \textbf{Collapse} all unmarked edges of $T$ and \textbf{run} Algorithm \ref{alg:int} on the residual graph.
        \Statex\Comment{\parbox[t]{.4\linewidth}{\textcolor{blue}{Algorithm \ref{alg:int} finds minimum exact-2-respecting cut when $T$ is a path}}}
        \State \textbf{Record} the output.
    \EndIf
\EndFor \label{algline:pairing-e}
\State \textbf{Output} the smallest of the recorded minimums.
\end{algorithmic}
\end{algorithm}
\end{minipage}
\end{center}




\begin{claim} \label{clm:tree-amortize}
Each edge in $T$ takes part in at most $O( \log n)$ many calls to Algorithm \ref{alg:int} inside Algorithm \ref{alg:tree}.
\end{claim}

\begin{proof}
For an edge $e$, the set of edges which are cross-\interest falls on a root-to-leaf path. By Theorem \ref{thm:tree-decomp}, there can be at most $\log n$ many paths which are cross-\interest to $e$. 

When $e' \in u^\da$ where $e$ is the parent of $u$, the set of edges which are down-\interest to $e$ falls to a $u$-to-leaf path, and by Theorem \ref{thm:tree-decomp} there can be at most $\log n$ many paths which are down-\interest to $e$. For $e'$, however, the argument is simpler: The set of $e$ which are down-interested in $e'$ are, by definition, falls on the root-to-$e'$ path---they are all ancestors of $e'$ in $T$. Hence, the number of paths which can pair up with the path containing $e'$ in calls to Algorithm \ref{alg:int} is at most $\log n$ (by Theorem \ref{thm:tree-decomp}).
\end{proof}

\subsubsection{Correctness of Algorithm \ref{alg:tree}}

First we make the following simple observation:
\begin{observation}
Consider two vertices $u$ and $v$ such that $u \bot v$. The 2-respecting cut which respects the parent edges of $u$ and $v$ is a candidate for exact-2-respecting min-cut (\textit{i.e.,} has cut value smaller than any 1-respecting cut) if $C(u^\da, v^\da) > \frac 1 2 \max\{\deg(u^\da), \deg(v^\da)\}$. Similarly, when $v \in u^\da$, the 2-respecting cut which respects the parent edges of $u$ and $v$ is a candidate for 2-respecting min-cut if $C(v^\da, V - u^\da) > \frac 1 2 \deg(u^\da)$.
\end{observation}

This follows from the fact that, for the 2-respecting cut which respects the parent edges of $u$ and $v$, when $u \bot v$, to be candidate for 2-respecting min-cut, it needs to happen that $\deg(u^\da) + \deg(v^\da) - 2 C(u^\da, v^\da) < \min\{\deg(u^\da), \deg(v^\da)\}$ where the right-hand side of the strict inequality represents the value corresponding 1-respecting cuts which respects the parent edges of $u$ and $v$ respectively. Similarly, when $v \in u^\da$, clearly, $\min\{\deg(u^\da), \deg(v^\da)\} = \deg(v^\da)$. Hence, it needs to happen that $\deg(u^\da) + \deg(v^\da) - 2 C(v^\da, V - u^\da) < \deg(v^\da)$.

\begin{claim}
All candidate exact-2-respecting cuts are considered in Algorithm \ref{alg:tree}.
\end{claim}

\begin{proof}
First, note that, any pair of edges $e$ and $e'$ such that (i) $e$ is cross-\interest to $e'$ and vice versa, or (ii) $e'$ is down-\interest to $e$ are considered. For any edge $e$, all $e' \bot e$ which are candidates are also cross-\interest to $e$ (by Definition \ref{def:cross-int-vertex}). Any $e' \in u^\da$ ($e$ is the parent edge of $u$) which is a candidate is also down-\interest to $e$ (be Definition \ref{def:down-int-vertex}).
\end{proof}

\section{Karger's tree packing} \label{sec:karger}

Here we first define the notion of greedy tree packing and cut-sparsification. This exposition is based on \cite{Thorup07} to which readers are advised to refer for a more detailed description.

\begin{definition}[(Greedy) tree packing]
A tree packing $\cal T$ of $G$ is a multi-set of spanning trees of $G$. $\cal T$ loads each edge $e \in E(G)$ with the number of trees in $\cal T$ that contains that edge $e$.

A tree packing ${\cal T} = (T_1, \cdots, T_k)$ is greedy if each $T_i$ is a minimal spanning tree with respect to the loads introduced by $\{T1, \cdots, T_{i-1}\}$.
\end{definition}

\begin{lemma}[\cite{Kar00}] \label{lem:karger-packing}
Let $C$ be any cut with at most $1.1 \lambda$ many edges and $\cal T$ be a greedy tree packing with $\lambda \ln m$ many trees. Then $C$ 2-respects at least $1/3$ fraction of trees in $\cal T$.
\end{lemma}

\begin{definition}[Cut sparsifier] \label{def:sparsifier}
Given a graph $G = (V,E)$ with weight function $w:E \to \bbR$, a sparsifier $H = (V, E')$ of $G$ is graph on the same set of vertices $V$ and weight function $w': E' \to \bbR$ with the following properties:
\begin{enumerate}
    \item $H$ has $\tilde O(n)$ edges,
    \item For every cut $S$, $C_{H}(S) \in (1 \pm \eps)C_{G}(S)$.
\end{enumerate}
\end{definition}

Now we are ready to provide the complete algorithm for finding min-cut in a weighted graph $G$. There are two phases: (i) tree-packing phase, where we compute a cut-sparsifier and pack appropriately many spanning trees, and (ii) cut-finding phase, where we find a 2-respecting min-cut which respects a randomly sampled tree from the first phase. The schematic description of the algorithm is given below. 

\begin{center}
  \centering
  \begin{minipage}[H]{0.8\textwidth}
\begin{algorithm}[H]
\caption{Schematic algorithm for weighted min-cut}\label{alg:tree-packing}
\begin{algorithmic}[1]
\State Compute a cut-sparsifier $H$ for $G$. \label{algline:tree-pack-s}
\State Pack $O(\lambda \ln m)$ many spanning trees $\cal T$ by greedy tree packing in $H$. \label{algline:tree-pack-e}
\State Pick a spanning tree $T \in \cal T$ uniformly at random and run Algorithm \ref{alg:tree}.
\end{algorithmic}
\end{algorithm}
\end{minipage}
\end{center}

The correctness follows from the fact that, by Definition \ref{def:sparsifier} in $H$, min-cut $\lambda' \in (1 \pm \eps)\lambda$. Hence, by Lemma \ref{lem:karger-packing}, if we greedily pack $O(\lambda \ln m)$ many trees, the min-cut of $H$ (and hence the min-cut of $G$) will 2-respect at least $1/3$ fraction of the packed trees. So, if we pick a tree uniformly at random, the sampled tree will be 2-respected by the min-cut with probability at least $1/3$.

In all three models that we describe subsequently, we will compute the cut-sparsifier $H$ efficiently and store it locally. The greedy packing of spanning trees will be performed also locally, and then we will run an algorithm for finding a 2-respecting min-cut (aka model specific implementation of Algorithm \ref{alg:tree}) on a randomly sampled tree from the set of packed trees.

In order to reducing the value of the min-cut of a graph (and, thereby, reducing the number of trees that are needed to be packed by a greedy tree packing algorithm), we use the following skeleton construction due to Karger \cite{Kar00}.

\begin{theorem}[\cite{Kar00}, Theorem 4.1] \label{thm:karger-skeleton}
 Given a weighted graph $G$, we can construct a skeleton graph $H$ of $G$ such that \begin{enumerate}
     \item $H$ has $m' = O(n \eps^{-2}\log n)$ edges,
    \item The minimum cut of $H$ is $O(\eps^{-2}\log n)$,
    \item The minimum cut in $G$ corresponds (under the same vertex partition) to a $(1 + \eps)$-times minimum cut of $H$.
 \end{enumerate}
 Thus, a set of $O(\log n)$ spanning trees can be packing in $G$ (by performing greedy tree packing on $H$) such  such that the minimum cut of $G$ 2-respects $1/3$ of them with high probability.
\end{theorem}

Note that Theorem \ref{thm:karger-skeleton} is useful only when the value of min-cut is $\omega(\log n)$. Otherwise, we can greedily pack trees in $G$ (or its sparsifier) itself to obtain a set of $O(\log n)$ spanning trees.

\section{Cut-query \& streaming algorithms} \label{sec:model-implement}

In this section, we provide model-specific implementation of Algorithm \ref{alg:tree-packing} in two models: In Section \ref{sec:q-implement} we look at graph-query model, and, in Section \ref{sec:s-implement}, we provide an algorithm in the dynamic streaming model. In all these sections, the arguments follow the following general pattern.
\begin{itemize}
    \item We first show that the randomized reduction from finding a weighted min-cut to finding a 2-respecting min-cut w.r.t to a given spanning tree can be implemented efficiently (Line \ref{algline:tree-pack-s} to \ref{algline:tree-pack-e} of Algorithm \ref{alg:tree-packing}).[Theorem \ref{thm:rand-sample} and Claim \ref{clm:s-reduction}]
    
    \item Next we show that we can implement Algorithm \ref{alg:int} efficiently. This algorithm will be called many times in Algorithm \ref{alg:tree}, and hence we want to make sure that this step is efficient. [Claim \ref{clm:q-interval}, and Claim \ref{clm:s-interval}]
    
    \item Then we show that each edge $e \in T$ can find the sets $\cP^{cross}_e$ and $\cP^{down}_e$ efficiently. This is a crucial step because this reduces the search space of 2-respecting cuts, and we want to make sure that this step can be done as quickly as possible. [Lemma \ref{lem:q-pairing}, and Claim \ref{clm:s-pairing}]
    
    \item Lastly, we analyse the other steps of Algorithm \ref{alg:tree} to conclude that the algorithm is efficient. [Lemma \ref{lem:q-tworespect}, and \ref{lem:s-tworespect}]
\end{itemize}

\subsection{Graph cut-query upper bound} \label{sec:q-implement}

In the graph cut-query model, we have oracle access to $G$ in the following way: We can send a partition $(S,\bar S)$ of the vertex set $V$ of $G$ to the oracle and the oracle will reply with the value of the cut corresponding to the given partition. The goal is to minimize the number of such accesses to compute weighted min-cut. Graph cut-queries can be quite powerful as the following claim suggests. Most of these are stated in \cite{RubinsteinSW18} and are fairly easy to verify.
\begin{claim}
In $O(\log n)$ cut-queries we can \begin{enumerate}
    \item learn one neighbor of a vertex $u \in V$ or a set of vertices $U \subset V$, 
    \item sample a random neighbor of a vertex $u \in V$ or a set of vertices $U \subset V$ and
    \item after performing $n$ initial queries, sample a random edge $e \in E$.
\end{enumerate}
And in 3 cut-queries, we can find out the total weight of all edges going between two disjoint set of vertices.
\end{claim}

Now, we state a result of \cite{RubinsteinSW18} regarding computing a cut-sparsifier efficiently. Note that, such a sparsifier is computed by using all of the above powers of cut-queries.

\begin{theorem}[\cite{RubinsteinSW18}] \label{thm:rand-sample}
Fix any $\eps > 0$. By using at most $\tilde O(n/\eps^2)$ cut-queries, we can produce a \textit{sparsifier} $H$ of $G$ such that:\begin{enumerate}
    \item $H$ has $O(n \ln n/\eps^2)$ edges,
    \item Every cut in $H$ is within a $(1 \pm \eps)$-factor of its value in $G$.
\end{enumerate}
\end{theorem}






At this point, we can perform greedy tree packing in $H$ and can sample a random spanning tree from the packed set of trees. We are left with the job of finding a min-cut which 2-respects the sampled tree\footnote{Note that we have find 2-respecting min-cut for $O(\log n)$ many sampled trees from the packed set of tress in order to attain very high success probability.}. Now we turn towards implementing Algorithm \ref{alg:tree} in the cut-query model. At first we show that Algorithm \ref{alg:int} is efficient.  As noted in Claim \ref{clm:interval-path-eqv}, any efficient query protocol for the Interval problem immediately implies an efficient query protocol for the 2-respecting min-cut problem where the underlying spanning tree is a path. In the Interval problem, we have oracle access to the entries of the cost-matrix $M$: This is because, given any pair of edges $e_i$ and $e_j$ of the spanning tree, we can find out the value of the 2-respecting cut which respects $e_i$ and $e_j$ by a single cut-query---this is the value of $\cost(i,j)$ in the Interval problem. But, even if we want to check all 2-respecting cuts, we need $\Omega(n^2)$ queries which we cannot afford. The following observation follows from the discussion so far.

\begin{observation}
Any entry of the cost-matrix $M$ associated with the Interval problem can be known in a single cut-query.
\end{observation}

This immediately implies that we can execute Line \ref{algline:min-find} of Algorithm \ref{alg:bipartite-int} with $O(n)$ cut-queries. So, using Claim \ref{clm:interval-gen-comp}, we can claim the following:

\begin{claim} \label{clm:q-interval}
The cut-query complexity of Algorithm \ref{alg:int} on a $n$ vertex graph is $O(n)$.
\end{claim}

Now we come to the implementation of Algorithm \ref{alg:tree}. We make the following claim next.

\begin{lemma} \label{lem:q-pairing}
We can find out $\cP^{cross}_e$ and $\cP^{down}_e$ for all $e \in T$ with $O(n)$ cut-queries.
\end{lemma}

\begin{proof}
 Let $\cC$ denote the set of cut queries we are interested in which are (i) $\cC(u^\da)$ for all $u \in V$, (ii) $\cC(u^\da, v^\da)$ for all pairs of distinct $u, v$ such that $u \bot v$, and (iii) $\cC(v^\da, V - u^\da)$ for all pairs of $u,v$ such that $v \in u^\da$. We can assume that \begin{enumerate}
    \item we have made $n$ many initial cut-queries to find out $\deg(u^\da)$ for each $u\in V$, and
    \item we also know the edges of the sparsifier $H$ from the tree-packing phase (Line \ref{algline:tree-pack-s} to \ref{algline:tree-pack-e} of Algorithm \ref{alg:tree-packing}).
\end{enumerate} 
Consider an edge $e'$ which is cross-\interest to $e$, i.e., $\deg(u^\da) < 2 C(u^\da, v^\da)$ where $e$ is the parent edge of $u$ and $e'$ is the parent edge of $v$. Now, by Theorem \ref{thm:rand-sample}, we know that in $H$, $\deg_H(u^\da) \leq (1 + \eps)\deg_G(u^\da) < 2(1 + \eps)C_G(u^\da, v^\da)$. 

To connect $C_G$ with $C_H$, we have work a bit more: Let us partition the vertex set $V$ of $G$ in three parts: $u^\da$, $v^\da$ and $R = V - (u^\da \cup v^\da)$. Now, we know that $2C_G(u^\da, v^\da) > \deg_G(u^\da) = C_G(u^\da, v^\da) + C_G(u^\da, R)$. Similarly, $2C_G(u^\da, v^\da) > C_G(u^\da, v^\da) + C_G(u^\da, R)$. Combining, we get $2 C_G(u^\da, v^\da) > C_G(u^\da, R) + C_G(v^\da, R)$. Now, let us look at the sparsifier $H$. We have the following:
\begin{align*}
    &2 C_H(u^\da, v^\da) \\
    &\underbrace{\geq}_{(\ast)} (1 - \eps)(C_G(u^\da, v^\da) + C_G(u^\da, R)) + (1 - \eps)(C_G(u^\da, v^\da) + C_G(v^\da, R)) - (1 + \eps) (C_G(v^\da, R) + C_G(u^\da, R))\\
    &= 2 (1 - \eps)C_G(u^\da, v^\da) - 2\eps (C_G(u^\da, R) + C_G(v^\da, R))\\
    & \underbrace{>}_{(\ast \ast)} 2 (1 - \eps)C_G(u^\da, v^\da) - 4 \eps C_G(u^\da, v^\da)\\
    & = 2 C_G(u^\da, v^\da)(1 - 3 \eps),
\end{align*}
where $(\ast)$ follows from the sparsifier guarantee, and $(\ast \ast)$ follows from what we observed before. Hence, we have
\[
\deg_H(u^\da) < 2(1 + \eps) C_G(u^\da, v^\da) < 2\frac{1 + \eps}{1 - 3 \eps}C_H(u^\da, v^\da).
\]
For small enough $\eps$, we have $\deg_H(u^\da) < 3 C_H(u^\da, v^\da)$. 

Hence, to find out which edges are cross-\interest w.r.t. $e$ in $G$, we need to check whether $C_H(u^\da, v^\da) >  \deg_H(u^\da)/3$ in $H$.  Extending Definition \ref{def:cross-int-vertex}, let us call the edge $e'$ which is parent of $v$ to be $H$-cross-\interest w.r.t. $e$ if $C_H(u^\da, v^\da) >  \deg_H(u^\da)/3$. Note that the set of $H$-cross-\interest edges w.r.t $e$ is superset of the set of edges which are cross-\interest w.r.t. $e$. By a similar argument as that of Observation \ref{obs:int-vbotv}, we know that there can be at most 2 root-to-leaf paths (none of which  contains $e$) in $H$ which can contain edges that are $H$-cross-\interest to $e$. There two root-to-leaf path will intersect with at most $2 \log n$ paths from $\cP$, and hence we need to check $C(u^\da, v^\da)$ for at most $2 \log n$ many vertices $v \bot u$ to find out which paths are actually cross-interesting to $e$. We can make $6\log n$ cut queries---3 queries for each path $p \in \cP$ which intersects these two root-to-leaf paths---in the original graph to figure out which paths in $\cP$ are interesting w.r.t. $e$, \textit{i.e.}, the set $\cP^{cross}_e$. A similar argument can be made for the set $\cP^{down}_u$.
\end{proof}

Now we analyze, as before, the cut-query complexity of Algorithm \ref{alg:tree}

\begin{lemma} \label{lem:q-tworespect}
The cut-query complexity of finding a 2-respecting weighted min-cut is $\tilde O(n)$.
\end{lemma}

\begin{proof}
We count the numner of cut-queries required at each line of Algorithm \ref{alg:tree}.
\begin{description}
\item[Line \ref{algline:1-resp-s} to \ref{algline:1-resp-e}:] This requires $n-1$ many cut-queries, one for each edge $e$. 
\item[Line \ref{algline:2-resp-path-s} to \ref{algline:2-resp-path-e}:] By Claim \ref{clm:q-interval}, each call to Algorithm \ref{alg:int} requires size of the path many cut-queries. The paths in $\cP$ are disjoint, and hence each edge takes part in exactly 1 path in $\cP$. Hence, the total number of queries required is $n-1$.
\item[Line \ref{algline:find-interest-s} to \ref{algline:find-interest-e}:] This requires the knowledge of the sparsifier graph $H$, which we can assume to possess because of the tree packing steps. This also requires the knowledge of $\deg(u^\da)$ for every $u$. We already know this from Line \ref{algline:1-resp-s} to \ref{algline:1-resp-e}. By Lemma \ref{lem:q-pairing}, this step can be executed with $\tilde O(n)$ many queries.

\item[Line \ref{algline:pairing-s} to \ref{algline:pairing-e}:] As before, we can use Claim \ref{clm:tree-amortize} and \ref{clm:q-interval} to conclude that this step requires $\tilde O(n)$ many cut-queries.
\end{description}
Hence, in total, $\tilde O(n)$ many cut-queries are required.
\end{proof}

\subsection{Streaming upper bound} \label{sec:s-implement}

In the dynamic streaming model, we assume that the algorithm knows the vertex set $V$ of $G$ and has access to a stream of edge-insertion and edge-deletion instructions: each such instruction declares a pair of vertices of the graph and an optional weight, and mentions whether an edge of the corresponding weight needs to be inserted between the two vertices, or if the edge which is already present between those two vertices should be removed. First we state a result of \cite{KLMMS17} regarding efficient computability of sparsifier.
\begin{theorem}[\cite{KLMMS17}] \label{thm:s-sparsifier}
 There exists an algorithm that processes a list of edge insertions and deletions for a weighted graph $G$ in a single pass and maintains a set of linear sketches of this input in $\tO(n)$ space. From these sketches, it is possible to recover, with high probability, a (spectral \footnote{A stronger notion than that of cut-sparsifier.}) sparsifier $H$ with $\tO(n)$ edges.
\end{theorem}

\begin{claim} \label{clm:s-reduction}
A dynamic streaming algorithm can perform the randomized reduction from weighted min-cut problem to a 2-respecting weighted min-cut problem in a single pass with $\tO(n)$ amount of total memory.
\end{claim}

\begin{proof}
Given a weighted graph $G$, the algorithm first uses Theorem \ref{thm:s-sparsifier} to construct a sparsifier $H$ of size $\tO(n)$ in a single pass. Next, the algorithm locally constructs a skeleton graph $H'$ of $H$ such that the value of the min-cut in $H'$ becomes $O(\log n)$. To this this, the algorithm uses Theorem \ref{thm:karger-skeleton} by setting $\eps$ to be a very small constant less than 1 (say $\eps = 1/100$). Now, the algorithm can perform greedy tree packing on $H'$ locally. Let us denote, as before, the greedy tree packing by ${\cal T} = \{T_1, \cdots, T_k\}$. Because of the skeleton construction, $k = O(\log n)$ and the algorithms can store the set $\cal T$ in its internal memory.\footnote{The skeleton construction was not necessary in the case of the cut-query model because of its assumption of unbounded internal computational power.} The algorithm does reservoir sampling \cite{Vit85} to pick a random tree $T$ from $\cal T$: This way the algorithm can sample a tree $T$ uniformly at random from $\cal T$ with $\tO(n)$ memory. By Lemma \ref{lem:karger-packing}, the min-cut 2-respects $T$ with constant probability.
\end{proof}

Now we turn towards implementing Algorithm \ref{alg:tree} in the dynamic streaming model. We first note the following: For a given cut, we can find the value of the cut in a dynamic stream with just $O(\log n)$ bits of memory---we maintain a counter and whenever the stream inserts (or deletes) an edge in (or from) the cut, we increase (or decrease) the counter. So the following observation is immediate.

\begin{observation} \label{obs:s-cut-value}
Any given $\tO(n)$ cut queries can be answered by a dynamic streaming algorithm in a single pass with $\tO(n)$ space.
\end{observation}

As before, by Claim \ref{clm:interval-path-eqv}, this means that the streaming algorithm can find $\tO(n)$ entries of the cost-matrix corresponding to the Interval problem in a single pass and $\tO(n)$ space. We prove that this is enough to implement Algorithm \ref{alg:int} efficiently.
\danupon{Use Misra-Dries majority sampling for all tree edges and all layers. One needs to bough decomposition for it.}
\begin{claim} \label{clm:s-interval}
A dynamic streaming algorithm can perform Algorithm \ref{alg:int} on an $n$ vertex graph in $O(\log n)$ pass with $\tO(n)$ bits of memory.
\end{claim}

\begin{proof}
We first look at the pass and memory requirement for performing Algorithm \ref{alg:bipartite-int}. Note that we can use Claim \ref{clm:interval-gen-comp} where we assume that a streaming algorithm can find $i_s$ and $i_t$ in Line \ref{algline:min-find} of Algorithm \ref{alg:bipartite-int} can be done in a single pass with $\tO(n)$ bits of memory (by Observation \ref{obs:s-cut-value}). Hence, by a similar calculation, it is easy to see that Algorithm \ref{alg:bipartite-int} can be performed in $O(\log n)$ passes and with $\tO(n)$ memory.

For Algorithm \ref{alg:int}, we can run $O(\log n)$ many occurrences of previous algorithm simultaneously, one for each level of recursion: The $i$-th algorithm will run $2^{i-1}$ many instances of {\sc Bipartite-interval}, each one on a path (or line) with $n/2^{i-1}$ edges (or points). Clearly, each algorithm can run in $O(\log n)$ passes with $\tO(n)$ memory, and hence the complexity of the combined algorithm is $O(\log n)$ passes with $\tO(n)$ memory as well.
\end{proof}

Now, as in the graph cut-query model, we turn towards the implementation of \Cref{alg:tree}. We make the following claim.

\begin{claim} \label{clm:s-pairing}
A dynamic streaming algorithm can find $\cP^{cross}_e$ and $\cP^{down}_e$ for all $e \in T$ in two passes and $\tO(n)$ memory.
\end{claim}

\begin{proof}
The proof is similar to Lemma \ref{lem:q-pairing}. As noted in Theorem \ref{thm:s-sparsifier}, the algorithm can compute a sparsifier $H$ of $G$ in a single pass with $\tO(n)$ memory. Once done, by a similar argument as that of Lemma \ref{lem:q-pairing}, each edge $e \in T$ needs to check $O(\log n)$ paths from $\cP$ to figure out which paths are cross (down)-interesting w.r.t. $e$ in the original graph $G$. This requires checking cut-values of $\tilde O(n)$ many cuts in total. This, by Observation \ref{obs:s-cut-value}, can be done in a single pass with $\tO(n)$ bits of memory.
\end{proof}


Now we analyze, as before, the streaming complexity of Algorithm \ref{alg:tree}.

\begin{lemma} \label{lem:s-tworespect}
There exists an algorithm that processes a list of edge insertions and deletions for a weighted graph $G$ on $n$ vertices can perform Algorithm \ref{alg:tree} in $O(\log n)$ passes and with total memory of $\tO(n)$ bits.
\end{lemma}

\begin{description}
\item[Line \ref{algline:1-resp-s} to \ref{algline:1-resp-e}:] There are $n-1$ cuts to check, and hence, by Observation \ref{obs:s-cut-value}, it can be done in a single pass with $\tO(n)$ bits of memory.

\item[Line \ref{algline:2-resp-path-s} to \ref{algline:2-resp-path-e}:] We know, by Claim \ref{clm:s-interval}, that the interval problem on $\ell$ many points can be performed in $O(\log \ell)$ passes and $\tO(\ell)$ bits of memory. As the set $\cP$ is edge-disjoint, we can run $|\cP|$ many streaming algorithm in parallel, all of which read the same dynamic stream. This needs $O(\log n)$ passes and $\tO(n)$ memory.

\item[Line \ref{algline:find-interest-s} to \ref{algline:find-interest-e}:] By Claim \ref{clm:s-pairing}, this can be done in 2 passes and with $\tO(n)$ memory. Note that, the first pass of Claim \ref{clm:s-pairing} is dedicated to finding a sparsifier $H$, which will be done in the tree-packing phase (Line \ref{algline:tree-pack-s}) of Algorithm \ref{alg:tree-packing} (See proof of Claim \ref{clm:s-reduction}). So, while implementing Algorithm \ref{alg:tree-packing}, these lines can be executed in a single pass instead of two passes.

\item[Line \ref{algline:pairing-s} to \ref{algline:pairing-e}:] We will run one instance of the dynamic streaming algorithm for each call to Algorithm \ref{alg:int} in parallel, all of which reads the same stream. This needs $O(\log n)$ many passes. For the memory bound, Claim \ref{clm:tree-amortize} guarantees that each edge is going to be used by at most $O(\log n)$ many such algorithms. As stated before, any such algorithm solving Interval problem on $\ell$ many points can be performed in $O(\log \ell)$ passes and $\tO(\ell)$ bits of memory. So these steps can be performed in $O(\log n)$ passes and $\tO(n)$ total memory.
\end{description}
Hence, Algorithm \ref{alg:tree} can be implemented by a dynamic streaming algorithm in $O(\log n)$ passes and $\tO(n)$ memory.


\section{Sequential algorithm for 2-respecting min-cut} \label{sec:seq}

In the following section, we provide a randomized algorithm for the problem of finding minimum 2-respecting min-cut. The error probability of this algorithm is polynomially small, \textit{i.e.,} $n^{-c}$ for any large enough constant $c$---we denote this as \textit{with high probability}. Our main result is the following. 

\begin{theorem}\label{thm:centralize-main}
There is a randomized algorithm that, given a spanning tree $T$ of an $n$-node $m$-edge (weighted) graph $G$, can find the 2-respect min-cut with high probability in $O(m \log n +n\log^4 n)$ time. If the graph is unweighted (but can have multiple edges between a pair of vertices), then the 2-respecting min-cut can be found in $O(m \sqrt{\log n} +n\log^4 n)$.
\end{theorem}



\Cref{thm:centralize-main} relies on a certain \textit{2-d range-counting and sampling data-structure} that we state in the following section. For the 2-d range counting data-structure and all other necessary data-strucures introduced in Section \ref{sec:seq-proof}, we will also look at the time complexity averaged out over a sequence of data-structure operations, \textit{i.e.,} we are interested in \textit{amortized} time complexity of data-structure operations.

\subsection{Range counting, sampling \& searching} \label{sec:range-count}

We start this section by introducing data-structure for 2-d orthognal range counting/sampling. We later introduce data-structure for 2-d semigroup range searching.

\begin{definition}[2-d Orthogonal Range Counting/Sampling]
An (online) range counting data structure is an algorithm that take $m$ points in the plane to pre-process in time $t_p$ and support a sequence of the following query operations: \begin{enumerate}
    \item given an axes-aligned rectangle, the algorithm outputs the number of points in the rectangle in time $t_c$ \textit{(range counting)}, 
    \item given an axes-aligned rectangle and an integer $k = O(\log m)$, the algorithm outputs $k'$ distinct random points in $R$ ($k \leq k' = O(k)$) with high probability in time $(k'+1)t_s$ \textit{(range sampling)}.
\end{enumerate}
We denote such a data-structure by $(t_p, t_c, t_s)$-data-structure. We will also be interested in amortized \textit{range reporting} time $t_r$: given an axes-aligned rectangle and an integer $k = O(\log m)$, the algorithm outputs $k'$ distinct points in $R$ ($k \leq k' = O(k)$) in time $(k'+1)t_r$.
\end{definition}

Below we present two  $(t_p, t_c, t_s)$-data-structures---one where we do not allow bit operations and one where we allow it---with their corresponding pre-processing, counting and sampling time.

\begin{theorem} \label{thm:seq-data-structure}
There are $(t_p, t_c, t_s)$-data-structures when given $m$ points in a 2-d plane to pre-process have the following complexities:
\begin{center}
    \begin{tabular}{ |p{3cm}|p{4cm}|p{3.8cm}|  }
 \hline
 \textbf{(Amortized) time} & \textbf{without bit-operations} & \textbf{with bit-operations}\\
 \hline
 $t_p(m)$ & $O(m\log m)$ & $O(m \sqrt{\log m})$\\
 $t_c(m)$ & $O(\log m)$ & $O\left(\frac{\log m}{\log \log m}\right)$\\
 $t_s(m)$ & $O(\log m)$ & $O\left(\frac{\log m}{\log \log m}\right)$\\
 \hline
\end{tabular}
\end{center}

\end{theorem}

\bsni
For proving \Cref{thm:seq-data-structure}, when bit-operations are not allowed, we use the following result by \cite{Chaz88}.

 \begin{theorem}[\cite{Chaz88}] \label{thm:chazellle}
 There is a 2-d range count/report data-structure which can pre-process $m$ points in time $O(m \log m)$ and answers range count/report query in $O(\log m)$ amortized time which does not require bit operations.
 \end{theorem}
 
 \begin{remark}
 Chazelle \cite{Chaz88} proved an even stronger result where he showed that reporting $k$ many points takes time $O(k + \log m)$. Also, as we will see in Section \ref{sec:seq-proof}, it is enough for these data-structures to work on 2-d grid (instead of 2-d plane). For this special case, Overmars \cite{Overmars88} has shown the existence of a data-structure with $t_p(m) = O(m \log m)$ and reporting $k$ points take time $O(k + \sqrt{\log m})$ (given that the grid is of size $m \times m$).
 \end{remark}

When we allow bit operations, we can do better. We use the following 1-d range rank/select query data-structure from \cite{BGKS15}: A range rank/select query data-structure, given $m$ numbers in an array $A[1, \cdots, m]$ to pre-process, can support the following types of queries:\begin{enumerate}
    \item a \textit{rank query} is associated with two indices $i,j \in [m]$ and an integer $I$, and outputs the number of integers in $A[i\cdots, j]$ that are smaller than an integer $I$,
    \item a \textit{select query} is associated with two indices $i, j \in [m]$ and a number $k$, and outputs the $k$-th smallest number in $A[i, \cdots, j]$.
\end{enumerate} 
The following two theorems are stated for data-structure where bit operations are allowed.

\begin{theorem}[\cite{BGKS15}]\label{thm:babenko}
There is a 1-d range rank/select query data-structure which can pre-process $m$ points in  time $O(m \sqrt{\log m})$ and answers range rank/select query in time $O(\frac{\log m}{\log \log m})$.
\end{theorem}

\begin{theorem}[\cite{CP10}]\label{thm:chan-patrascu}
There is a 2-d orthogonal range counting data-structure which can pre-process $m$ points in time $O(m \sqrt{\log m})$ and answers range count query in time $O(\frac{\log m}{\log \log m})$.
\end{theorem}

\begin{proof} [Proof of Theorem \ref{thm:seq-data-structure}]
First we show how to construct such a data-structure when bit-operations are not allowed. The main idea is to show how to construct a range-sampling data-structure from a range-reporting data-structure from \Cref{thm:chazellle}. Subsequently, we show how to construct a more efficient range-reporting data-structure when bit-operations are allowed. We will show how to construct an efficient range-reporting data-structure from \Cref{thm:babenko} and \ref{thm:chan-patrascu}.

\paragraph*{Without bit-operations.} We show the proof by constructing the (randomized) data-structure and analysing its correctness. We create a data-structure $D$ as in Theorem \ref{thm:chazellle}. The time required for pre-processing $m$ points is $t_p(m) = O(m \log m)$ and the time required for range counting is $t_c(m) = O(\log m)$. The rest of the proof is to show $t_s(m) = O(\log m)$ as well.
    
    To this end, we show that, given a 2-d range-reporting data-structure with pre-processing time $t_p(m) = O(m \log m)$ and range-reporting time $t_r$, we can construct a range-sampling data-structure with pre-processing time $t_p(m) = O(m \log m)$ and $t_s = t_r$. The proof then concludes by observing that $t_r(m) = O(\log m)$ for $D$ by \Cref{thm:chazellle}. The reduction is described as follows:
    
    \begin{itemize}
	\item {\bf Pre-processing:}
	We define the following set of points $S_0,  S_1, \ldots,  S_k$, where $k=\log m$ in the following way: Let $S_0 = S$. For $i$ starting from 1 to $k,$ we define $S_i \subseteq S_{i-1}$ by copying each point of $S_{i-1}$ to $S_i$ with probability $1/2$.  For every $0\leq i\leq k$, we also do the following: We build a 2-d range reporting data structure $D_i$ for each point set $S_i$. This takes time $\sum_{i=0}^k t_p(|S_i|)$ where $t_p$ is the pre-processing time of the range reporting data structure. This sums up to  $\sum_{i=0}^k O(|S_i|\log |S_i|)=O(m\log m)$.

	\item {\bf Query($R, k$):} Let $i=k$. We query the range reporting data structure $D_i$ for points in $R\cap S_i$ and output the reported points. If less than $c\log n$ points were reported so far, we repeat with $i\leftarrow i-1$. We stop when all points in $R$ are reported (when $i=0$).  $D_i$ reports each point in  $t_r(|S_i|) = O(\log |S_i|)$ amortized time. Note that even when $R$ contains no point, we need an extra $O(\log m)$ time to iterate over all values of $i$. (This can be avoided using the range counting data structure, but this refinement does not affect our overall running time so we do not include it here.)\danupon{Perhaps things can be stated easier if we also use the counting data structure here?}
	
	\end{itemize}

\noindent\underline{Analysis:} For any $i\leq k$, let $S_{\geq i}=\bigcup_{j\geq i} S_j$. Observe that, for every $i$, every point is in $S_{\ge i}$ with probability $1/2^i$. Since we answer the query by reporting all points in   $S_{\ge i}$ for some $i$, every point is output with the same probability. Now we bound the number of points we output. Let $m'$ be the number of points in $R$ in total. Let $c\geq 2$ be a large enough constant whose value is independent from the input size. Let $i'$ be such that $m'/2^{i'}\in (ck, 2ck]$. The expected number of points in $R$ that is in $S_{\geq i'}$ is $E[|R\cap S_{\geq i'}|]=m'/2^{i'}\geq ck\geq c\log m$; so by the standard Chernoff's bound we have that $|R\cap S_{\geq i'}|\in [ck/2, 4ck]$  with (very) high probability. If this is the case, the query algorithm will stop when it reports points in $R\cap S_{\geq i'}$ at the latest, implying that it reports at most  $|R\cap S_{\geq i'}|=O(k)$ points. 	\danupon{TO DO: Spell out details of Chernoff's bound?}	
The proof follows by noting $t_r(m) = \max_i O(\log|S_i|) = O(\log m)$.

\paragraph*{With bit-operations.} Similar to what we did previously, we show the proof by constructing the (randomized) data-structure and analysing its correctness. We will create a data-structure $D$ as in Theorem \ref{thm:babenko} and another data-structure $D'$ as in \Cref{thm:chan-patrascu}. The time required for pre-processing for both $D$ and $D'$ is $O(m \sqrt{\log m})$. This immediately gives $t_p(m) = O(m \sqrt{\log m})$. Also, for range counting, we use $D'$. By \Cref{thm:chan-patrascu}, $t_c(m) = O(\frac{\log m}{\log \log m})$. Previously, we have seen that, given a range-reporting data-structure with reporting time $t_r$, we can design a range-sampling data-structure with sampling time $t_s = t_r$. A similar argument also holds here as well\footnote{One has to be careful about pre-processing time which is now $\sum_i O(|S_i|\sqrt{\log |S_i|}) = O(m \sqrt{\log m})$ as $\sqrt{ \log(x)}$ is an increasing function w.r.t. $x$.}. Hence, the rest of the proof is to show that $t_r(m) = O\left(\frac{\log m}{\log \log m}\right)$.

	Concretely, we show how to construct a range-reporting data-structure with pre-processing time $t_p(m) = O(m \sqrt{\log m})$ and range-reporting time $t_r = O(\frac{\log m}{\log \log m})$ from $D$ and $D'$. Note that $D$ is a 1-d data-structure. Hence we need to map the points 2-d orthogonal plane to a single dimension before $D$ can pre-process it. For this, we create an array $\cA$ of size $m$. For points $e = (e^{(x)},e^{(y)})$ in the 2-d plane, where $e^{(x)}$ and $e^{(y)}$ are the $x$ and $y$ coordinates of $e$, we order them with respect to the $x$-coordinate, and for points with same $x$-coordinate, we order them by $y$-coordinate. Given this ordering, we have a one-to-one correspondence between $m$ points in the 2-d plane and the entries of $\cA$: The entry $\cA[i]$ corresponds to an unique point $e_i$ in the 2-d plane. We store the $y$-coordinate of $e_i$ in $\cA[i]$, i.e., $\cA[i] = e_i^{(y)}$.
	
	Now we demonstrate how to report one point in a given axes-aligned rectangle $R$. Let us denote the columns of $R$ by $\{R_1, \cdots, R_t\}$ where $R_i$ resides in the column $C_i$ of the 2-d plane. We also consider two more (related) axes-aligned rectangles, and for that we need to introduce the following notion of \textit{dominance}. A point $e$ in the 2-d plane \textit{dominates} another point $e'$ if $e^{(x)} \geq e'^{(x)}$ and $e^{(y)} \geq e'^{(y)}$. Extending this notion, a rectangle $R_1$ \textit{dominates} a disjoint rectangle $R_2$ if each point in $R_1$ dominates all points in $R_2$. The two other rectangles we consider are as follows: (i) Let $\tilde R$ be the ambient rectangle formed by taking union of columns $\{C_1, \cdots, C_k\}$, and (ii) $\bar R$ to be the biggest sub-rectangle of $\tilde R$ which is dominated by and disjoint from $R$. The following two observations are immediate: (i) The points in the rectangle $\tilde R$ corresponds to a contiguous range in $\cA$, and (ii) w.r.t. the ordering of $\cA$, the entries in $\cA$ corresponding to the points of $\bar R$ are less than that of $R$, i.e., $\cA[i] < \cA[j]$ if $e_i \in \bar R$ and $e_j \in R$. \begin{enumerate}
	    \item We query $D'$ to find the count of points in $\bar R$. Let the number of such points is $c$.
	    
	    \item We also query $D'$ to find the count of points in $R \cup \bar R$. Let this count be $c'$. This means that $c'-c$ many points reside in $R$. If $c' - c = 0$, we give up and return 0.
	    
	    \item Otherwise, we query $D$ in the range corresponding to $\tilde R$ to select the smallest $(c+j)$-th entry of $\tilde R$ for all $j \leq c'-c$. Note that these points must belong inside $R$ if $c' - c \neq 0$. Finding the range corresponding to $\tilde R$ can be achieved by making 2 count-queries to $D'$: one to find out the number of points on the left side of $\tilde R$ which gives us the starting point of the range; and the other to find out the number of points on the right side of $\tilde R$ which gives us the ending point of the range.
	\end{enumerate} 
Hence, to report $k$ points in $R$, the total query time is $O\left((k + 4)\frac{\log m}{\log \log m}\right)$. This gives an amortized query time of $t_r(m) = O(\frac{\log m}{\log \log m})$ as required. This completes the proof of \Cref{thm:seq-data-structure}.
\end{proof}

\begin{remark} \label{rem:multiset} 
Even though the above proof of Theorem \ref{thm:seq-data-structure} can deal with multiple points in a single coordinate of the 2-d plane, we do not need such strong guarantee in our work. In what follows, we will have at most a single point in a single coordinate of the 2-d plane.
\end{remark}

Next we introduce semigroup range searching data-structure. This is an extension of range counting data-structure where we can assume that each point in the 2-d plane has a weight associated with it, and given an axes-aligned rectangle we want to know that total weight of the points inside the rectangle (instead of counting just the number of points inside it).

\begin{definition}[2-d Semigroup Range Searching]
Consider a function $f$ from points on a 2-d plane to a commutative semi-group $(G, +)$. An (online) semi-group range searching data structure is an algorithm that take $m$ points in the plane to pre-process in time $t_p^{\sg}$ and support a sequence of the following query operation: \begin{itemize}
    \item given an axes-aligned rectangle $R$, the algorithm outputs $\sum_{e \in R}f(e)$ where the summation is the semigroup operation.
\end{itemize}
We denote the query time as $t_c^{\sg}$, and call such a data-structure as $(t_p^{\sg}, t_c^{\sg})$-data-structure.
\end{definition}

 We will make use of the following  $(t_p^{\sg}, t_c^{\sg})$-data-structure (for handling weighted graph). This data-structure does not require bit operations.

\begin{theorem}[\cite{Chaz88}] \label{thm:chazellle-semigroup}
 There is a  $(t_p^{\sg}, t_c^{\sg})$-data-structure in which can pre-process $m$ points in time $t_p^{\sg}(m)=O(m \log m)$ and answers range search query in $t_c^{\sg}(m)=O(\log^2 m)$ amortized time.
 \end{theorem}

\subsection{Main algorithm: Proof of Theorem \ref{thm:centralize-main}} \label{sec:seq-proof}

In this section, we will assume the existence of a $(t_p, t_c, t_s)$-data-structure as in \Cref{thm:seq-data-structure} and a $(t_p^{\sg}, t_c^{\sg})$-data-structure as in \Cref{thm:chazellle-semigroup}, and will prove \Cref{thm:centralize-main}. First, we show a few operations we can perform on a graph which are immediate from such data-structures.

\begin{lemma}\label{thm:seq:convert to point}
For any spanning tree $T$ of an $n$-node $m$-edge unweighted graph $G$, there is an algorithm that, after a $t_p(m)$-time pre-processing, can answer the following queries in $t_c(m)$ amortized time:  
\begin{enumerate}
    \item for any vertex $v$, $\deg(v^\da)$,
    \item for any two vertices $u$ and $v$, $C(u^\da, v^\da)$, and
    \item for any two vertices $u$ and $v$, $C(u^\da, V - v^\da)$.
\end{enumerate}
When $G$ is weighted, then these operations need $t_p^{\sg}(m)$-time for pre-processing and $t_c^{\sg}(m)$-time to answer.
\end{lemma}

\begin{proof}
We sequentialize the vertices in $V$ by the trace of a post-order traversal of $T$. Let us denote this order by $\prec$. We will also use this notation to define $S \prec T$ for set of vertices $S$ and $T$ if all vertices in $S$ occur before any vertex of $T$ in the ordering. A post-order traversal order has the following property: For any vertex $v$ with children $u_1 \prec \cdots \prec u_k$, the ordering guarantees that $u_1^\da \prec \cdots \prec u_k^\da$. This immediately implies that \begin{enumerate}
    \item for all $u \in V$, $u^\da$ is a continuous range in the ordering, and
    \item for all $v \in V$, $V - v^\da$ comprises of two continuous ranges in the ordering.
\end{enumerate}
So we plot $V$ in this order along both axes, and, for all edge $e$ in the graph $G$, put the weight of the edge $e=(u,v)$ at the coordinate $(u,v)$ on the plane. (We need to account for $e$ only once between coordinates $(u,v)$ and $(v,u)$.) By formulating $\deg(v^\da) = C(v^\da, V - v^\da)$, we see that each of the cut-queries in the claim can be answered by at most 2 range-count queries to the $(t_p, t_c, t_s)$-data-structure if the edges have weight 1 each, or at most 2 queries to the $(t_p^{\sg}, t_c^{\sg})$-data-structure if the edges are weighted.
\end{proof}

Now we look at Algorithm \ref{alg:bipartite-int} and \ref{alg:int}, and see how they can be implemented efficiently if we have a $(t_p, t_c, t_s)$-data-structure. We start with the following lemma.

\begin{lemma} \label{lem:seq-pairing}
The exists an algorithm that supports the following operations. 
\begin{itemize}
    \item Pre-process($T, G$): The algorithm is given a spanning tree $T$ of an $n$-node $m$-edge unweighted graph $G$ to pre-process. 
    \item Query($A$, $B$): The algorithm is given two disjoint lists (arrays) of edges in $T$, denoted by $A=\{a_1, \ldots, a_x\}$ and $B=\{b_1, \ldots, b_y\}$, where $a_1, \dots, a_x$ (respectively  $b_1, \ldots, b_y$) appear in a path in $T$ in order; i.e. $a_1$ appears first in the path and $a_x$ appears last. (Note that, for any $i$, $a_i$ might not be adjacent to $a_{i+1}$.) The algorithm then outputs $\min_{(a_i, b_j)\in A\times B} \cost(a,b)$. 
\end{itemize}
After $t_p$ pre-processing time, the algorithm takes $O(\ell \log \ell \cdot t_c(m))$ amortized time to answer each query, where $\ell=|A|+|B|$. For weighted graph, the pre-processing time is $t_p^{\sg}$, and query time is $O(\ell \log \ell \cdot t_c^{\sg}(m))$.
\end{lemma}

\begin{proof}
The algorithm follows the schematic of Algorithm \ref{alg:bipartite-int}. The only difference is that the set $A$ and $B$ are not contiguous sequence of edges. But we can pretend the edges which does not take part in $A$ or $B$ as \textit{collapsed}. We first look at the case where the graph is unweighted.
\begin{enumerate}
    \item The algorithm prepares the data-structure from Lemma \ref{thm:seq:convert to point},
    \item Given any query $(A,B)$, the algorithm runs Algorithm \ref{alg:bipartite-int}. In line \ref{algline:min-find} of Algorithm \ref{alg:bipartite-int}, it needs to issue a total of $O(x)$ many 2-respecting cut queries (\textit{i.e.,} queries of the form $C(u^\da, v^\da)$) in each level of recursion. From Theorem \ref{thm:seq:convert to point}, such queries can be answered in $t_c(m)$ amortized time. Hence, in each level of recursion, the algorithm needs time $O(x \cdot t_c(m))$. There are at most $\log x$ levels of recursion, which implies the total running time is $O(x \log x \cdot t_c(m))$.
\end{enumerate}
As $\ell > x$ we have the required running time. The case of weighted graph can be handled similarly by using $(t_p^{\sg}, t_c^{\sg})$-data-structure from \Cref{thm:seq:convert to point}.
\end{proof}

We look at a related lemma for implementing the {\sc Interval} sub-routine, but this time with contiguous set of edges.
\begin{lemma} \label{lem:seq-2-resp-path}
There exists an algorithm which, given an unweighted graph $G$ and and a spanning tree $T$, can pre-process in time $t_p(m)$ and can answer the following question: Given a path $p \in T$ of length $\ell$, the algorithm can output cut-value the smallest exact-2-respecting cut which respects two edges from  in time $O(\ell \log^2 \ell \cdot t_c(m))$. For weighted graph, the pre-processing time is $t_p^{\sg}(m)$, and query time is $O(\ell \log^2 \ell \cdot t_c^{\sg}(m))$.
\end{lemma}

\begin{proof}
This algorithm follows the schematic of Algorithm \ref{alg:int}. As before, we can pretend that edges of $T$ not in $p$ to be \textit{collapsed}. For the sake of efficiency, the algorithm stores the edges of $p$ in sequence in an array. Also, as before, we first look at the case of unweighted graphs.
\begin{enumerate}
\item The algorithm  prepares the data-structure from Lemma \ref{thm:seq:convert to point},
\item Given any path $p \in T$, the algorithm runs Algorithm \ref{alg:int} on $p$. The algorithm needs to issue $O(\ell \log^2 \ell)$ many exact-2-respecting cut queries (\textit{i.e.,} queries of the form $C(u^\da, v^\da)$) to the data-structure. Hence the total running time is $O(\ell \log^2 \ell \cdot t_c(m))$.
\end{enumerate}
The case of weighted graphs can be handled similarly by using $(t_p^{\sg}, t_c^{\sg})$-data-structure from \Cref{thm:seq:convert to point}.
\end{proof}

Now we look at the most crucial part of our sequential implementation: How to find the set $\cP^{cross}_e$ and $\cP^{down}_e$ efficiently for each $e \in T$? To this end, we introduce a random sampling lemma below. But, before stating the lemma, we need to an assumption on the edge weights of $G$, which is as follows:
	

\begin{lemma}[Initial assumption]\label{lem:seq:max weight}\label{lem:seq:assumption}
By spending linear time, we can assume that the input graph $G$ has no parallel edge and the maximum edge weight is at most $3\lambda$ where $\lambda$ is the min-cut of $G$. 
\end{lemma}
\begin{proof}
We can assume, without loss of generality, that $G$ has no parallel edge simply by merging parallel edges into one edge.
As pointed out by Karger \cite{Karger99-skeleton,Karger_thesis}, one can modify Matula's algorithm \cite{Matula93} to get a $(2+\epsilon)$-approximate value of $\lambda$ in linear time. So we can assume that we have $\lambda'\in [\lambda, 3\lambda]$. We then contract all edges of weight more than $\lambda'$. Observe that the min-cut does not change after the contraction. (Note: Karger \cite{Karger99-skeleton} also described a simple $n^2$-approximation algorithm, which can be used to prove a weaker lemma than here. Such lemma is also enough for our purpose.)
\end{proof}

\begin{lemma}[Finding big $C(u^\da, v^\da)$] \label{lem:seq:estimate interesting}
There is a data structure that, after pre-processing a weighted $m$-edge connected graph $G=(V, E)$ and a spanning tree $T$ of $G$ in $O(t_p(m))$ time, can answer the following type of queries: 
\begin{itemize}
	\item Query($u$): Given a node $u$, the algorithm returns $O(\log n)$ sets of edges in $G$, denoted by $F_1(u), F_2(u), \ldots$, such that with high probability 
	\begin{itemize}
		\item[(I)] $|F_i(u)|=O(\log n)$ for all $i$, and 
		\item[(II)] for any node $v$ where $C(u^\da, v^\da) > \deg(u^\da)/2$, there exists $i$ such that the number of edges in $F(u)$ that connect between $u^\da$ and $v^\da$, \textit{i.e.,} $|F_i(u)\cap \cC(u^\da, v^\da)|$ is at least $|F_i(u)|/16$. 
	\end{itemize}

\end{itemize}
The algorithm takes $t_s(m) \sum_i|F_i(u)|$ amortized time to answer each query.
\end{lemma}


 \begin{proof}
	
	{\em Intuition:} If the graph is unweighted, we can simply sample $\Theta(\log n)$ edges from $\deg(u^\da)$ using the $(t_p, t_c, t_s)$-data structure from \Cref{thm:seq-data-structure} (where we convert edges to points as in  \Cref{thm:seq:convert to point}). We can then use the Chernoff's bound to argue that if $C(u^\da, v^\da) > \deg(u^\da)/2$ then one-third of the sample edges are between  $u^\da$ and $v^\da$ (where we view both $u^\da$ and $v^\da$ as ranges in two orthogonal axes as in Theorem \ref{thm:seq:convert to point}). 
	For weighted graphs, more has to be done, as follows.
	
	\paragraph*{Pre-processing:} Recall that, by \Cref{lem:seq:max weight}, we can assume that the maximum edge weight is at most $3\lambda$ and there is no parallel edge.
	Let $G'$ be a graph obtained by removing all edges in $G$ of weight less than $\epsilon\lambda/n^2$, for a small enough constant $\epsilon$. Thus, the ratio between the minimum and maximum weight of $G'$ is $3n^2/\epsilon$. To simplify the exposition, we assume that the minimum weight is $1$. 
	Partition edges of $G'$ into $E_0, E_1, \ldots, E_\ell$ where $E_i$ is the set of edges whose weights are in the range $[2^i, 2^{i+1})$. (Note that $\ell=O(\log n)$.) Convert each edge in $G'$ into a 2-dimensional point by a post-order traversal in $T$ as in \Cref{thm:seq:convert to point}. Let $P_i$ be the set of points resulting from converting edges in $E'_i$. Let $D_i$ be a $(t_p, t_c, t_s)$-data structure that is given points in $P_i$ to pre-process\footnote{Even though the edges in $P_i$ have weights in the range $[2^i, 2^{i+1})$, $D_i$ treats all edges equally by disregarding edge weights. More informally, $D_i$ treats the subgraph induced by $P_i$ as an unweighted simple graph. See Remark \ref{rem:multiset} for relevance.}.

	\paragraph*{Query(u):} For every $i$, we query $D_i$ to sample $k=c\log n$ many points, for a big enough constant $c$, where the queried range is defined to cover edges between $u^\da$ and $V(G)\setminus u^\da$ (which can be done according to \Cref{thm:seq:convert to point}.) Let $F_i(u)$ denote the set of edges corresponding to points returned by $D_i$. In other words, $F_i(u)$ is a set of random edges in between $u^\da$ and $V(G)\setminus u^\da$ in $E'_i$.
	
	
	
	\paragraph*{Analysis:} 
	By \Cref{thm:seq-data-structure}, the processing requires $\sum_i t_p(|P_i|) = O(t_p(m))$ time and the query outputs $O(\ell k)=O(\log^2 n)$ edges with high probability. The number of points sampled is $\sum_i |F_i(u)|$, and hence the time required to sample these many points is $t_s(m) \sum_i |F_i(u)|$. This is the time required to answer Query($u$). Additionally these data structures require $O(\log n)$ amortized time to output each point, which becomes an output edge. 
	
	Now it is left to show (II).
	That is, if we let $v$ be any node such that $C_G(u^\da, v^\da) > \deg_G(u^\da)/2$, then for some $i$
	\begin{align}
	|F_i(u)\cap \cC(u^\da, v^\da)|&\geq |F_i(u)|/16.\label{eq:seq:estimateGoal}
	\end{align}
	%
	Let $w(E')$ be the total weight of any set of edges $E'$.
	Note that $w(E(G)\setminus E(G'))\leq \epsilon\lambda$. 
	Since $\deg_{G'}(u^\da)\geq \lambda$ for any node $u$ (since it defines a cut), we have
	%
	\begin{align*}
	C_{G'}(u^\da, v^\da)  &> (1-2\epsilon)C_G(u^\da, v^\da) > (1-2\epsilon) \deg_G(u^\da)/2.
	\end{align*}

	By the averaging argument we have that there is some $i$ such that $w(E_i\cap 	\cC_{G'}(u^\da, v^\da))>\frac{1-2\epsilon}{2}w(E_i\cap \cC_G(u^\da)).$ Since all edges in $E_i$ have weights within a factor of two from each other, we have 
	$$|E_i\cap 	\cC_{G'}(u^\da, v^\da)| > \frac{1-2\epsilon}{4}|E_i\cap \cC_G(u^\da)|.$$
	\danupon{Need to spell out details?}
	Thus, when we sample $k=c \log n$ edges from $E_i\cap \deg_G(u^\da)$ to construct $F_i(u)$, the expected number of edges in $E_i\cap 	\cC_{G'}(u^\da, v^\da)$ is at least $\frac{1-2\epsilon}{4} c \log n$. So, by Chernoff's bound, we have
	$$|F_i(u)\cap 	\cC_{G'}(u^\da, v^\da)| > \frac{1-2\epsilon}{8}|F_i(u)|$$
	with high probability when $c$ is large enough. With small enough $\epsilon$, \Cref{eq:seq:estimateGoal} follows.
\end{proof}

We have a similar claim for down-\interest edges. The proof is similar to that of Lemma \ref{lem:seq:estimate interesting} and is omitted.

\begin{lemma}[Finding big $C(v^\da, V - u^\da)$] \label{lem:seq:estimate interesting_down}
There is a data structure that, after pre-processing a weighted $m$-edge connected graph $G=(V, E)$ and a spanning tree $T$ of $G$ in $O(t_p(m))$ time, can answer the following type of queries: 
\begin{itemize}
	\item Query($u$): Given a node $u$, the algorithm returns $O(\log n)$ sets of edges in $G$, denoted by $F'_1(u), F'_2(u), \ldots$, such that with high probability 
	\begin{itemize}
		\item[(I)] $|F'_i(u)|=O(\log n)$ for all $i$, and 
		\item[(II)] for any node $v$ where $C(v^\da, V - u^\da) > \deg(u^\da)/2$, there exists $i$ such that the number of edges in $F'(u)$ that connect between $u^\da$ and $v^\da$, \textit{i.e.,} $|F'_i(u)\cap \cC(v^\da, V - u^\da)|$ is at least $|F'_i(u)|/16$. 
	\end{itemize}
\end{itemize}
The algorithm takes $t_s(m)\sum_i |F'_i(u)|$ amortized time to answer each query, which is $O(\log^2 n \cdot t_s(m))$ with high probability.
\end{lemma}

Next, we need a data-structure for heavy-light decomposition which, on an edge query, will provide the set of paths of decomposition which intersects the root-to-leaf path containing that edge. We use the \textit{compressed tree} data-structure from \cite{HT84}. A minor modification of this data-structure will give the following lemma.

\begin{lemma}[\cite{HT84}]\label{lem:seq-decomp-ds}
Consider $\cP$ to be a set of paths obtained by heavy-light decomposition of a tree $T$. Given a tree $T$ on $n$ vertices, there is a data-structure which can pre-process the tree in $O(n)$ times such that it can answer the following type of queries:
\begin{itemize}
\item Query($e$): Given an edge $e$ of the graph, the algorithm returns at most $\log n$ many edges $e'_1, \cdots, e'_k$ and a set of paths $\{p_1, \cdots, p_k\} \subseteq \cP$ such that (i) each $e'_i$ belongs to the same root-to-leaf path that $e$ belongs to, (ii) each $e'_i$ belongs to a distinct path $p_i$ from $\cP$, and (iii) $e'_i$ is the edge closest to the root in $p_i$.
\end{itemize}
The algorithm takes $O(\log n)$ time to answer each such query.
\end{lemma}

We need one more data-structure for technical purpose.
\begin{lemma} \label{lem:seq-b-tree}
Consider a family $\cS$ of $n$ many sets of a universe. There is a data-structure which allows the following operations: \begin{enumerate}
\item \textbf{Insertion.} Given a tuple $(S,T,x)$ where $S$ and $T$ are sets in $\cS$ and $x$ is an element in $S$, the data-structure inserts $(S,T, x)$ in time $O(\log n)$.
\item \textbf{Sequential access.} At any point, provides sequential access to a set of tuples of the form $(S,\{x_1, \cdots, x_s\},$ $T, \{y_1, \cdots, y_t\})$ in $O(1)$ amortized time where, for each $x_i$, there was an insertion $(S,T, x_i)$ and, for each $y_i$, there was an insertion $(T,S, y_i)$.
\end{enumerate}
\end{lemma}

\begin{proof}[Proof sketch.]
The data-structure maintains an inherent ordering of the sets in $\cS$. Given any tuple $(S, T, x)$, where $S \prec T$ w.r.t. that ordering, the data-structure first searches whether there is any entry corresponding to $(S,T)$ in it. This can be searched in $O(\log n)$ time by binary-search. If there is such an entry, the data-structure adds $x$ to the set corresponding to $S$ in that entry. Otherwise, the data-structure creates an entry corresponding to $(S,T)$ at the lexicographically correct position (w.r.t. the ordering of $\cS$) and adds $x$ to the set corresponding to $S$ in that entry. When $T \prec S$, the data-structure works exactly in the same way, but does the search with the tuple $(T,S)$ instead of $(S,T)$ for not creating double entry. This data-structure can be very efficiently maintained by traditional data-structure, such as, B+ tree which also provides the quick sequential access as claimed.
\end{proof}

At this point, we look at the implementation of Algorithm \ref{alg:tree} in the sequential model. The algorithm, given a weighted graph $G$ and a spanning tree $T$, will construct a few data-structures to efficiently perform the computations, but the schematics will remain more-or-less the same. We explain the steps as follows:

\begin{description}
\item[Initial data-strucure creation:] The algorithm first creates a data-structure $D_1$ as in \Cref{thm:seq:convert to point}. This can be done in time $t_p^{\sg}(m)$. For unweighted graph, the time needed is $t_p(m)$. The algorithm also initializes a min-heap $\heap$.

\item[Line \ref{algline:1-resp-s} to \ref{algline:1-resp-e}:] For each edge $e$ of $T$, the algorithm issues query $\deg(u^\da)$ to $D_1$ where $e$ is the parent edge of $e$, and inserts it in \heap. This can be done in time $O(n \cdot t_c^{\sg}(m))$ for weighted graphs, and $O(n \cdot t_c(m))$ for unweighted graphs.

\item[Line \ref{algline:heavy-light}]: This can be done in time $O(n)$. At this point, the algorithm also creates a data-structure $D_2$ as in Lemma \ref{lem:seq-decomp-ds}.

\item[Line \ref{algline:2-resp-path-s} to \ref{algline:2-resp-path-e}:] The algorithm uses $D_1$, as in Lemma \ref{lem:seq-2-resp-path}, to find out the minimum exact-2-respecting cut in each path $p \in \cP$ and inserts it into \heap. The total amount of time needed, by Lemma \ref{lem:seq-2-resp-path}, is $O(n \log^2 n \cdot t_c(m))$ for unweighted graphs, and $O(n \log^2 n \cdot t_c^{\sg}(m))$ for weighted graphs.

\item[Line \ref{algline:find-interest-s} to \ref{algline:find-interest-e}:] \begin{enumerate}
\item The algorithm first creates data-structures $D_3^{cross}$ as in Lemma \ref{lem:seq:estimate interesting} and $D_3^{down}$ as in Lemma \ref{lem:seq:estimate interesting_down}. The algorithm also initiates a data-structures $D_4^{cross}$ and $D_4^{down}$ as in Lemma \ref{lem:seq-b-tree}.

\item The algorithm iterates over each vertex $u \in V$. \begin{enumerate}
	\item The algorithm queries $u$ in $D_3^{cross}$ to obtain the set $\{F_i(u)\}_i$. Note that the number of edges in $\bigcup_i F_i(u)$ is at most $O(\log^2 n)$. This requires time $O(\log^2 n \cdot t_s(m))$.
	
	\item For each edge $e \in \bigcup_i F_i(u)$, let $v$ be the end-point of $e$ not in $u^\da$. The algorithm queries $v$  in $D_2$ to obtain a set of edges $\{e'_1, \cdots, e'_k\}$ of size $O(\log n)$ along with the paths $\{p'_1, \cdots, p'_k\}$ where $e'_i \in p'_i$. Hence, in total, $O(\log^3 n)$ many edges corresponding to $e$ are collected in this process---we call this set $E^T_e$. Each query takes time $O(\log n)$, and hence in total $O(\log^3 n)$ time is required.
	
	\item The algorithm then inserts $O(\log^3 n)$ many tuples of the form $(p, p',e')$ in $D_4^{cross}$ where $e \in P$, $e' \in E^T_e$, and $p'$ is the path in $\cP$ which contains $e'$. Each insertion takes time $O(\log n)$, and hence $O(\log^4 n)$ time is required in total.
	
	\item The algorithm also follows Lemma \ref{lem:seq:estimate interesting_down} to query $D_3^{down}$ to obtain $O(\log n)$ many sets $\{F'_i(u)\}_i$, and executes exactly similar steps: For every edge $e \in \bigcup_i F'_i(u)$, the algorithm queries $D_2$ to obtain all potential interesting set of edges and path (up to $u$), and inserts all these $O(\log^3 n)$ tuples in $D_4^{down}$. The algorithm also inserts all tuples of the form $(p,p',e')$ in $D_4^{down}$ where $e \in p$, $e'$ is an ancestor of $e$ and $e' \in p'$. Total time required is, as before, $O(\log^4 n)$. 
	
	\item As the algorithm iterates over each edge $u \in V$, the total time required in $O(n (\log^4 n + \log^2 n \cdot t_s(m)))$.
\end{enumerate}
\end{enumerate} 

\item[Line \ref{algline:pairing-s} to \ref{algline:pairing-e}:] The algorithm uses $D_1$, as in Lemma \ref{lem:seq-pairing}, to find the minimum 2-respecting-cuts for each tuple in $D_4^{cross}$ and $D_4^{down}$. A similar amortized analysis as that in Section \ref{sec:q-implement} or \ref{sec:s-implement} will yield a running time of $O(n \log n \cdot t_c(m))$ for unweighted graphs, and $O(n \log n \cdot t_c^{\sg}(m))$ for weighted graph.
\end{description}

At the end, the algorithm finds the minimum entry from $\heap$ and outputs it. Let us compute the total running time for this implementation when the graph is weighted. The total running time of Algorithm \ref{alg:tree} is (ignoring lower order terms) 

$$O(t_p^{sg}(m) + t_p(m) + n \log^2 n \cdot t_s(m) + n \log^2 n \cdot t_c^{\sg}(m) + n \log^4 n),$$ 

which is $O(m \log n + n \log^4 n)$ if we plug in the value of $t_p^{\sg}(m), t_c^{\sg}(m), t_p(m), t_s(m)$ from \Cref{thm:seq-data-structure} and \Cref{thm:chazellle-semigroup}. This implementation does not use any bit operation. Note that we can use the range-sampling data-structure where bit-operations are allowed instead of the one where bit-operations are not allowed---Unfortunately this does not improve the running time as $t^{\sg}_p$ dominates $t_p$. In the unweighted case, however, we can replace the semigroup range searching data structure by the range counting data structure, and so we do not need to use the $(t_p^{\sg}, t_c^{\sg})$-data-structure. Hence the running time in this case is

$$O(t_p(m) + n \log^2 n \cdot t_s(m) + n \log^2 n \cdot t_c(m) + n \log^4 n),$$ 

which is $O(m \sqrt{\log n} + n \log^4 n)$ if we put values of $t_p(m), t_c(m), t_s(m)$ from Theorem \ref{thm:seq-data-structure} and \ref{thm:seq-data-structure}. Clearly, this case requires bit operations. 

We will now look at the correctness of this implementation. The most interesting part of this algorithm is Line \ref{algline:find-interest-s} to \ref{algline:find-interest-e}. From Lemma \ref{lem:seq:estimate interesting}, we know that if $v$ is cross-\interest to $u$ in $G$, then there is an $i$ such that $|F_i(u) \cap \cC(u^\da, v^\da)| \geq |F_i(u)|/16$. Let pick such an edge $e' \in F_i(u) \cap \cC(u^\da, v^\da)$ and let $v'$ be the vertex of $e'$ in $v^\da$. Note that $v$ belongs to the same root-to-leaf path as $v'$. Hence, all edges of the root-to-leaf path which are cross-\interest w.r.t. $e$ will be discovered by a query to $D_2$ with $v'$ \footnote{In fact, we do not need such a strong guarantee on $F_i(u) \cap \cC(u^\da, v^\da)$ for this purpose. We just need that $|F_i(u) \cap \cC(u^\da, v^\da)| \neq \emptyset$.}. As we iterate over all $i$, the set $\cP^{cross}_e$ will be discovered w.h.p. A similar argument can be made for $\cP^{down}_e$ using Lemma \ref{lem:seq:estimate interesting_down}. This immediately implies the following claim. 

\begin{claim}
In Line \ref{algline:find-interest-s} to \ref{algline:find-interest-e}, the algorithm finds, for every edge $e$, the set $\cP^{cross}_e$ and $\cP^{down}_e$.
\end{claim}

The correctness of the rest of the implementation follows from the correctness of the data-structures (\Cref{thm:seq:convert to point}, \ref{lem:seq-pairing} and \ref{lem:seq-2-resp-path}).

\section{Open problems}\label{sec:open}

In this work we show that weighted min-cut can be solved efficiently in many settings. One setting that remains difficult is for dynamic graphs. Improved dynamic algorithms (or conditional lower bounds) for maintaining min-cut exactly and approximately are both very challenging. For exact min-cut, the tree packing can be maintained using existing dynamic minimum spanning tree algorithms (e.g. \cite{HenzingerK99,HolmLT01,Wulff-Nilsen13a,KapronKM13,NanongkaiSW17}). The challenge is again maintaining the minimum 2-respecting cut. See \cite{Thorup07} for the state of the art of approximately maintaining the min-cut.
Additionally, it is interesting to see if the number of passes can be improved in the streaming setting.

All efficient min-cut algorithms are {\em randomized} (there is a near-linear time deterministic algorithm in the sequential setting, but it works only on simple graphs \cite{KawarabayashiT19}). Can we obtain an efficient {\em deterministic} algorithm in any of the settings considered in this paper? Note that our algorithm is randomized due to the need of a cut sparsifier. A cut sparsifier can be obtained deterministically in almost-linear time using the recent result in \cite{ChuzhoyGLNPS19}, but the approximation ratio is too high ($n^{o(1)}$). Improving the approximation ratio is the first important step.

Two problems related to min-cut, namely directed min-cut and vertex connectivity, are still widely open. We leave as a major open problem designing an efficient algorithm for any of these problems either in the sequential, streaming, and cut-query setting, or any other interesting computational model. One model that is of our interest is the two-party communication model where edges are partitioned into two parties. In this setting, the two party can compute the min-cut by communicating $\tilde O(n)$ bits. Can they use the same amount of communication for computing the minimum directed cut or vertex connectivity? In fact, this question is also open and interesting for other graph problems such as the maximum bipartite matching.

\section*{Acknowledgement} 

We thank Michal Dori and Sorrachai Yingchareonthawornchai for actively participating in the discussions during their visits at KTH, when we were trying to prove a communication complexity lower bound for the interval problem (which turns out to be impossible due to our results). We would also like to thank Yuval Efron, Joseph Swernofsky and Jan van den Brand for fruitful discussions on different stages of this work.

We thank Timothy Chan for pointing out relevant results on range searching data structures. We also thank Paweł Gawrychowski for pointing out an error in the earlier version, where we used Chazelle's range counting instead of semigroup range searching data structure. (This change does not affect our results.) 
\danupon{SAKNIK: Check}

This project has received funding from the European Research Council (ERC) under the European Union's Horizon 2020 research and innovation programme under grant agreement No 715672. Authors are also supported by the Swedish Research Council (Reg. No. 2015-04659 and 2019-05622).

\bibliography{biblio}

\appendix
\section{Proof of Claim \ref{clm:interval-path-eqv}} \label{app:interval-path-eqv}
A path $T$ of length $n$ has $n-1$ edges. We consider the edges $\{e_1, \cdots, e_{n-1}\}$ of the path $T$ from left to right. This ordering naturally defines left and right vertex of an edge in the path. Given this ordering, we identify each edge $e_\ell$ with the point $\ell$ in the Interval problem. For an edge $e$ (which is not in the path) which connects the left vertex of edge $e_i$ and the right vertex of edge $e_j$, we identify $e$ with an interval $I_e= (i,j)$.

If we consider a 2-respecting cut that respects edges $e_i$ and $e_j$ ($i \leq j$), then we claim that the cost of this cut is equal to $\cost(i,j)$ in the Interval problem. This is because:
\begin{enumerate}
    \item For any edge $e$ to contribute to this cut, the edge has to connect a vertex $v$, which is between right-side of edge $e_i$ and left-side of edge $e_j$, with a vertex $u$ which is either on the left-side of edge $e_i$ or right-side of edge $e_j$. The interval $I_e$, will cover either $i$ or $j$, but nor both, and hence contributes to $\cost(i,j)$. In Figure \ref{fig:interval}, where we are interested in $\cost(2, 5)$, interval $I_3$ is such an interval.
    \item For any edge $e'$, which connects the left-side of edge $e_i$ to the right-side of edge $e_j$, $I_e$ covers both $i$ and $j$. For any edge $e''$, which is contained between the right-side of edge $e_i$ to the left-side of edge $e_j$, $I_{e'}$ covers none of $i$ and $j$. In both cases, the corresponding intervals do not contribute to $\cost(i,j)$. In Figure \ref{fig:interval}, interval $I_1$ is the first kind of interval, and $I_2$ is the second kind. Neither of $I_1$ and $I_2$ contributes to $\cost(2,5)$.
\end{enumerate}
Note that this argument goes both ways, \textit{i.e.,} $\cost(i',j')$ represents the cost of a 2-respecting min-cut which respects edge $e_{i'}$ and $e_{j'}$.
\end{document}

%% file: preamble.tex
\usepackage[bookmarks,colorlinks,breaklinks]{hyperref}
\hypersetup{urlcolor=blue, colorlinks=true, citecolor=green!50!black, linkcolor=blue}
\usepackage[letterpaper, left=1in, right=1in, top=0.9in, bottom=0.9in]{geometry}
\usepackage[utf8]{inputenc}
\usepackage[american]{babel}

\usepackage[normalem]{ulem}

\usepackage{graphicx}
\graphicspath{{images/}{../images/}}

\usepackage{amsmath, amssymb, cases}
\usepackage{theoremstyles}
\usepackage{cleveref}
\usepackage{enumitem}
\usepackage{mdframed}
\usepackage{bbm}
\usepackage{bm}

\usepackage[ruled]{algorithm}
\usepackage[]{algpseudocode}

\usepackage{microtype}

\usepackage{mdframed}
\usepackage{xcolor}

\input{commands}

%% file: commands.tex
\usepackage{modletters}

\usepackage{mathtools}
\DeclarePairedDelimiter{\ceil}{\lceil}{\rceil}

\newcommand{\Interest}{Interesting\ }
\newcommand{\interest}{interesting\ }
\newcommand{\heap}{\textsf{Heap}}
\newcommand{\sg}{\mathsf{sg}}

\newcommand{\cost}{\mathsf{Cost}}
\newcommand{\da}{\downarrow}

\newcommand{\bsni}{\bigskip\noindent}

\newcommand\restr[2]{{
  \left.\kern-\nulldelimiterspace 
  #1 
  \vphantom{\big|} 
  \right|_{#2} 
  }}

\newcommand{\eps}{\varepsilon}

\makeatletter
\def\moverlay{\mathpalette\mov@rlay}
\def\mov@rlay#1#2{\leavevmode\vtop{%
   \baselineskip\z@skip \lineskiplimit-\maxdimen
   \ialign{\hfil$\m@th#1##$\hfil\cr#2\crcr}}}
\newcommand{\charfusion}[3][\mathord]{
    #1{\ifx#1\mathop\vphantom{#2}\fi
        \mathpalette\mov@rlay{#2\cr#3}
      }
    \ifx#1\mathop\expandafter\displaylimits\fi}
\makeatother

\usepackage{graphicx}

\newcommand{\size}[1]{\mathrm{size}}

\newcommand{\set}[2][ ]{\{#2 \ifthenelse{\equal{#1}{ }}{ }{~|~#1}\}}

\newcommand{\comment}[1]{}

\newcommand{\seepage}[2][See]{
    \marginnote{
        \scriptsize {#1} p.~\pageref{#2}
    }
}

\newcommand{\reuse}[1]{
	\expandafter\stepcounter{#1_help}
    \expandafter\label{#1_app}
    \csname#1\endcsname*
}

%% file: abstract.tex

\begin{abstract}

Consider the following {\em 2-respecting min-cut} problem. Given a weighted graph $G$ and its spanning tree $T$, find the minimum cut among the cuts that contain at most two edges in $T$. This problem is an important subroutine in Karger's celebrated randomized near-linear-time min-cut algorithm [STOC'96]. We present a new approach for this problem which can be easily implemented in many settings, leading to the following randomized min-cut algorithms for weighted graphs.

\begin{itemize}
	\item An $O\left(m\frac{\log^2 n}{\log\log n} + n\log^6 n\right)$-time sequential algorithm: This improves Karger's long-standing $O(m \log^3 n)$ and $O\left(m\frac{(\log^2 n)\log (n^2/m)}{\log\log n} + n\log^6 n\right)$ bounds when the input graph is not extremely sparse or dense. Improvements over Karger's bounds were previously known only under a rather strong assumption that the input graph is {\em simple} (unweighted without parallel edges) [Henzinger, Rao, Wang, SODA'17; Ghaffari, Nowicki, Thorup, SODA'20]. For unweighted graphs  (possibly with parallel edges) and using bit operations, our bound can be further improved to $O\left(m\frac{\log^{1.5} n}{\log\log n} + n\log^6 n\right)$. 
	\danupon{SAGNIK: Check}
	
	\item An algorithm that requires $\tilde O(n)$ {\em cut queries} to compute the min-cut of a weighted graph: This answers an open problem by Rubinstein, Schramm, and Weinberg [ITCS'18], who obtained a similar bound for simple graphs. Our bound is tight up to polylogarithmic factors. 
	
	\item A {\em streaming} algorithm that requires $\tilde O(n)$ space and $O(\log n)$ passes to compute the min-cut: The only previous non-trivial exact min-cut algorithm in this setting is the 2-pass $\tilde O(n)$-space algorithm on simple graphs [Rubinstein~et~al., ITCS'18] (observed by Assadi, Chen, and Khanna [STOC'19]). 
\end{itemize}

\danupon{Below was changed after discussing with Bryce Sandlund}

Our approach exploits some cute structural properties so that it only needs to compute the values of $\tilde O(n)$ cuts corresponding to removing $\tilde O(n)$ pairs of tree edges, an operation that can be done quickly in many settings.  This is in contrast to the techniques used by Karger and Lovett-Sandlund to solve 2-respecting min-cut where information about many more cuts is computed, stored in and accessed from sophisticated data-structures.
\end{abstract}

%% file: intro.tex
\section{Introduction} \label{sec:intro}

\paragraph{Min-Cut.} Given a weighted graph $G$, a {\em cut} is a set of edges whose removal disconnects the graph. The {\em minimum cut} or {\em min-cut} is the cut with minimum total edge weight. The min-cut problem---finding the min-cut---is a classic graph optimization problem with countless applications. 
A long line of work spanning over many decades in the last century was concluded by the STOC'95 $O(m\log^3 n)$-time randomized algorithm of Karger \cite{Kar00}. (Throughout, we $n$ and $m$ denote the number of nodes and edges respectively.) Karger's approach remains the only approach to solve the problem in near-linear time for general graphs.  With more tricks, he improved the running time further to $O(m(\log^2 n)\log (n^2/m)/\log\log n + n\log^6 n)$), suggesting the possibility to improve the logarithmic factors further. He also suggested a few approaches that might improve the running time by a $\log n$ factor. Nevertheless, after more than two decades no further improvement was known.   


The only improvements over Karger's bound known to date are under a rather strong assumption that the input graph is {\em simple}, i.e. it is unweighted and contains no parallel edges. The first such bound was by the $O(m \log^2 n\log\log^2 n)$-time deterministic algorithm of Henzinger, Rao and Wang \cite{HenzingerRW17} (following the breakthrough $O(m\log^{12} n)$-time deterministic algorithm of Kawarabayashi and Thorup \cite{KawarabayashiT19}). More recently Ghaffari, Nowicki and Thorup \cite{GhaffariNT20} improved this bound further to $O(m \log n)$ and $O(m+n \log^3 n)$ with randomized algorithms. 

By restricting to simple graphs, progress has been also made in other settings: (I) In the {\em cut-query} setting, we are allowed to query for the value of a cut specified by a set of nodes. The goal is to compute the minimum cut value. Na\"ively, one can make $O(n^2)$ cut queries to reconstruct the graph itself. 
Rubinstein, Schramm and Weinberg \cite{RubinsteinSW18} showed a randomized algorithm that only needs $\tilde O(n)$ queries, but their algorithm works only on simple graphs; here $\tilde O$ hides polylogarithmic factors. 
(II) In the {\em multi-pass semi-streaming} setting, an algorithm with $\tilde O(n)$ space reads the input graph in multiple passes, where in each pass it reads one edge at a time in an adversarial order. 
%
Assadi, Chen and Khanna \cite{AssadiCK19} observed that the algorithm of Rubinstein \textit{et al.}~can be adapted to solve the problem in two passes.
When it comes to non-simple graphs, no non-trivial algorithms were known in both settings.

\paragraph{2-Respecting Min-Cut.} The main bottleneck in obtaining similar results on non-simple graphs is the lack of efficient algorithms for the {\em 2-respecting cut problem}. In this problem, we are given a spanning tree $T$ of $G$ and have to find the minimum cut in $G$ among the cuts that contain at most two edges in $T$ (we say that such cuts {\em 2-respect} $T$). 

Solving this problem is a core subroutine and the bottleneck in Karger's algorithm. Using tree packing, Karger proved that if $\Tres(m, n)$ is the time needed to find the 2-respecting min-cut, then the min-cut can be found in 
\begin{align}
\Tcut(m,n)=O(\Tres(m, n) \log n + m + n \log^3 n) \mbox{ and }\nonumber\\ \Tcut(m,n)=O\left(\Tres(m, n) \frac{\log n}{\delta\log\log n} + n \log^6 n + m \log^{1 + \delta} n\right)\label{eq:KargerCutFromTree}
\end{align}
time for any small $\delta > 0$.
His $O(m \log^3 n) $ bound for min-cut was obtained by using a sophisticated dynamic programming algorithm to find 2-respecting min-cut in $\Tres(m, n)=O(m\log^2 n)$ time. (A simplification of this algorithm is presented in \cite{LovettS19}.) With more tricks he showed that 
$\Tres(m,n)=O(m(\log n)\log (n^2/m))$, leading to the bound of  $\Tcut(m,n)=O\left(m\frac{\log^2 n\log (n^2/m)}{\delta\log\log n} + n\log^6 n + m \log^{1+ \delta} n\right)$. \footnote{It is to be noted that Karger's result states $\Tcut(m,n)=O(\Tres(m, n) \log n /\log\log n + n \log^6 n + m \log^2 n)$. But, with a more careful analysis, the $m \log^2 n$ factor in the $\Tcut$ expression can be made $m \log^{1 + \delta} n$ for arbitrarily small positive $\delta$ without severely affecting other terms of the expression. To achieve this, simply set the threshold for $\alpha$ to be $1/(4\log^\delta n)$ in the proof of Lemma 9.1 of \cite{Kar00}. It can be checked that the tree-packing can be done in time $O(n\log^6 n)$ with these new parameters.} 
%
%
Using efficient cut sparsifier algorithms that exist in the cut-query and streaming setting (e.g. \cite{AhnGM12-pods,AhnGM12-soda,GoelKP12,RubinsteinSW18}), it is also not hard to adapt Karger's arguments to the cut-query and streaming settings, leaving the 2-respecting min-cut as the main bottleneck.

\paragraph{Our Results.} 
We show a new way to solve the 2-respecting min-cut problem which can be turned into efficient algorithms in many settings. In the sequential setting, our algorithm takes $\Tres(m, n)=O(m\log n+n\log^4 n)$ time and is correct with high probability. For unweighted graphs (possibly with parallel edges), our bound can be improved to
$\Tres(m, n)=O(m\sqrt{\log n}+n\log^4 n)$ when bit operations are allowed. Plugging this in \Cref{eq:KargerCutFromTree}, we achieve the first improvement over Karger's bounds on general graphs. 


\begin{theorem}
	\label{thm:intro:seq}
A min-cut can be found w.h.p. in $O(m\frac{ \log^2 n}{\log \log n} + n \log^6 n)$ time. For unweighted graphs (possibly with parallel edges) and with bit operations, the bound can be improved to  $O(m\frac{ \log^{3/2} n}{\log \log n} + n \log^6 n)$.\footnote{Recall that with high probability (w.h.p.) means with probability at least $1-1/n^c$ for an arbitrarily large constant $c$.}
\end{theorem}

\setlength{\tabcolsep}{18pt}
\renewcommand{\arraystretch}{1.4}


\begin{table}[t]
\centering
    \begin{tabular}{ |p{3cm}| c | c|  }
 \hline
 \textbf{Reference} & \textbf{Complexity} & \textbf{Remark}\\
 \hline
 \multicolumn{3}{|c|}{$ m\ge n \log^4 n$}\\ 
 \hline
 \cite{Kar00} & $O(m\frac{ \log^2 n \log (n^2/m)}{\log \log n} + n \log^6 n)$ 
 & Old record\\
  Here,\cite{GMW19} & $O(m\frac{ \log^2 n}{\log \log n} + n \log^6 n)$ 
  & New record\\
  Here & $O(m \frac{\log^{3/2} n}{\log \log n} + n \log^6 n)$
  & Unweighted\\
 \hline
 \multicolumn{3}{|c|}{$m \leq n \log^4 n$} \\
 \hline
 \cite{Kar00} & $O(m\log^3 n)$ & Old record\\
\cite{GMW19} & $O(m \log^2 n)$ & New record\\
 \hline
  \multicolumn{3}{|c|}{Simple graphs (unweighted, no parallel edges)} \\
 \hline
\cite{GhaffariNT20} & $O(\min(m+n\log^3 n, m\log n))$ & Old record \\
\cite{GMW19} & $O(\min(m+n\log^2 n, m\log n))$ & New record\\
\hline
 \end{tabular}
\caption{Results for sequential min-cut} \label{tab:results-seq}
\end{table}

\begin{table}[t]
\centering
    \begin{tabular}{ |p{3.5cm}| c | c|  }
 \hline
 \textbf{Reference} & \textbf{Complexity} & \textbf{Type of graphs}\\
 \hline
 \multicolumn{3}{|c|}{Cut-query algorithm} \\
 \hline
 \cite{RubinsteinSW18} & $\tilde O(n)$ & Simple  \\
 Here & $\tilde O(n)$ & Weighted \\
 \hline
 \multicolumn{3}{|c|}{Dynamic streaming algorithm} \\
 \hline
 \cite{RubinsteinSW18, AssadiCK19} & $\tO(n)$-space, 2-pass & Simple \\
 Here & $\tilde O(n)$-space, $O(\log n)$-pass & Weighted \\
 \hline
\end{tabular}
\caption{Our results in cut-query and streaming model and comparison with other works} \label{tab:results-cut-stream}
\end{table}

\noindent{\em Independent work and comparisons:} The recent independent work of  Gawrychowski, Mozes, and Weimann \cite{GMW19} shows that the 2-respecting min-cut and global min-cut can be solved in $O(m \log^2 n)$ time. This also improves the $O(m+n\log^3 n)$ bound of Ghaffari et~al. \cite{GhaffariNT20} to $O(m+n\log^2 n)$. 
\Cref{tab:results-seq} compares Gawrychowski et~al.'s result with ours and Karger's. For most values of $m$, their and our results are similar and improve Karger's result. For unweighted graphs, our bound can be improved further slightly. 
When the graph is very sparse, Gawrychowski et~al.'s bound improves Karger's (our result does not). When $m=\Omega(n^2)$, both results do not improve Karger's result, except for unweighted graphs where we provide a $\sqrt{\log n}$ improvement. 
%

\danupon{SAGNIK: CHECK!} 




\medskip
It is quite easy to implement our algorithm to find a 2-respecting min-cut using $\tilde O(n)$ cut queries in the cut-query model, and $O(\log n)$ passes with $\tO(n)$ internal memory in the semi-streaming settings respectively. Extending these results to find a min-cut (using Karger's tree packing approach) does not increase any additional overhead. Note that our streaming algorithm works even in the {\em dynamic streaming} setting where there are both edge insertions and deletions. 


\begin{theorem}
	\label{thm:intro:query}
A min-cut can be found w.h.p. using $\tilde O(n)$ cut queries.
\end{theorem}

\begin{theorem}
	\label{thm:intro:stream}
A min-cut can be found w.h.p. by an $O(\log n)$-pass $\tilde O(n)$-space dynamic streaming algorithm.
\end{theorem}

Rubinstein \textit{et al.}~\cite{RubinsteinSW18} asked whether it is possible to compute the min-cut on weighted graphs with $o(n^2)$ cut queries. (They showed an $\tilde O(n)$ bound for simple graphs (unweighted without parallel edges).) Our cut-query bound answers their open problem positively. Using a connection to communication complexity pointed out by Rubinstein \textit{et al.}~\cite{RubinsteinSW18}, we can prove an $\tilde \Omega(n)$ lower bound for the number of cut queries, making our result in \Cref{thm:intro:query} tight up to polylogarithmic factors.\danupon{Sketch in appendix?}
The $\tilde O(n)$ space bound in \Cref{thm:intro:stream} is also tight \cite{FeigenbaumKMSZ05}. As mentioned earlier, the only non-trivial algorithm in the streaming setting (even with insertions only) is the $\tilde O(n)$-space $2$-pass algorithm \cite{RubinsteinSW18,AssadiCK19}. Note that there are some $(1+\epsilon)$-approximation $\tilde O(n)$-space  $1$-pass algorithms in the literature (e.g.~\cite{KLMMS17,AhnGM12-soda,AhnGM12-pods}). Computing min-cut exactly in a single pass with $\tilde O(n)$ space is however impossible \cite{Zelke11}. 

\paragraph{Updates since our work.} A recent result \cite{gawrychowski2021note} shows a deterministic implementation of our schematic algorithm for 2-respecting cut with running time $\Tres(m, n) = O(m \log n + n \log^2 n)$ in the sequential model. By slightly altering the \textit{range search} data-structure that we use in this paper (explained in the Section \ref{sec:tech_overview} and \ref{sec:range-count}), they also obtain a runtime of $\Tres(m, n) = O(\frac{m}{\eps}  + \frac{n^{1+2\eps}\log n}{\eps^2}+n\log^2 n)$.\footnote{This needs a small change from \cite{gawrychowski2021note}, where we need to store $T_S$ in \cite{gawrychowski2021note} as a compacted trie. We thank Pawe{\l} Gawrychowski for the clarification.} Below we summarize the state of the art in the sequential model when this result is taken into account. 

\begin{itemize}
    \item When $m=O(n\log^2 n)$, we have $O(m\log^2 n)$ runtime \cite{GMW19}. 
    \item When $m=\omega(n\log^2 n)$, we have $O(\frac{m \log n}{\epsilon}  + \frac{n^{1+2\eps}\log^2 n}{\eps^2}+n\log^3 n)$ runtime (this paper together with \cite{gawrychowski2021note}).\footnote{When $m=n\log^2(n)g(n)$ for some growing function $g$, we have $\frac{m \log n}{\epsilon}  + \frac{n^{1+2\eps}\log^2 n}{\eps^2}+n\log^3 n< m \log^2 n$ by setting $\epsilon$ to $c(\log\log(g(n))/\log n$ for some constant $c$.} 
\end{itemize}

Our framework has also been implemented in the distributed and parallel settings, leading to a near-optimal distributed algorithm in the CONGEST model \cite{DoryEMN21} and the first work-optimal parallel algorithm \cite{lopezmn21}. 

\subsection{Technical Overview}  \label{sec:tech_overview}
Our approach exploits some cute structural properties so that we can limit the number of pairs of tree edges we include in our search space. 
%
This is in contrast to the techniques used by Karger \cite{Kar00} and Lovett-Sandlund \cite{LovettS19} to solve 2-respecting min-cut where information about many more cuts is computed---either using dynamic programming or using heavy-light decomposition---and stored in and accessed from sophisticated data-structures (Lovett-Sandlund additionally show that augmented binary search trees suffice, at the cost of an extra $O(\log n)$ factor).

%
We now explain the main ideas behind our algorithm. Let $G$ and $T$ be the input graph and spanning tree respectively. 
For any pair of edges $e$ and $f$ in $T$, let \cut($e$,$f$) be the number of edges $(u, v)$ where the unique $uv$-path in $T$ contains {\em exactly one} of $e$ and $f$. It is well-known that if a 2-respecting min-cut $C$ contains $e$ and $f$, then it has weight $\cut(e,f)$.
This is because the end-vertices of $e$ has to be in different connected components when we remove edges in $C$; otherwise, we can get a smaller cut by removing $e$ from $C$. (The same applies for $f$.) 

%
%


Our algorithm exploits the fact that \cut($e$, $f$), for a pair of tree-edges $(e,f)$, can be computed efficiently in many settings. For example, in the streaming setting, we only need to make a pass over all edges using $O(\log n)$-bit space (in addition to the space to keep $T$).  We also require only one cut query since we know how nodes are partitioned (e.g. end-vertices of $e$ should be in the different connected components after the graph is cut).
In the sequential setting, we can use the {\em 2-d orthogonal range counting} data structure: We assign nodes with integers based on when they are visited in a post-order traversal on $T$, and map each edge to a two-dimensional point based on the numbers assigned to their end-vertices. It is not hard to see that we can compute \cut($e$, $f$) if we know the number of points in a few rectangles. This information can be provided by the 2-d orthogonal range counting data structure. Using Chazelle's semigroup range search data structure \cite{Chaz88}, this takes $O(m\log m)$ preprocessing time and $O(\log n)$ amortized time to compute each \cut($e$, $f$). With bit operations, the pre-processing time can be improved to $O(m\sqrt{\log m})$ using, e.g., Chan and Patrascu's range counting data structure \cite{CP10}.


By the above facts, we can na\"ively find the 2-respecting min-cut by computing  $\cut(e,f)$ for all $n^2$ tree-edge pairs $(e, f)$. 
By exploiting some structural properties explained below, we can reduce the number of pairs in our search space to $\tilde O(n)$.  Below we focus on the number of {\em probes}---the number of tree-edge pairs $(e, f)$ that we have to compute $\cut(e,f)$ for. 
%
Since all other steps can be performed efficiently in all settings we consider, we obtain our results.


\paragraph{When $T$ is a path.}
To explain the structural properties that we use, let us consider some extreme cases. The first case is when $T$ is a path. It is already very unclear how to solve this case efficiently. (In fact, the starting point of this work was the belief that this case requires $\tilde \Omega(n^2)$ probes, which is now proven wrong.) 
The key insight is the structure of the {\em special case} of the problem, where we assume that the tree-edge pair $(e,f)$ that minimizes \cut($e$, $f$) is on the {\em different sides} of a given node $r$. 
More precisely, assume that, in addition to the path $T$, we are given a node $r$. For convenience, we use $e'_i$ and $e_i$ to denote the $i^{th}$ edges of $T$ to the left and right of $r$ respectively (see the graph in \Cref{fig:bipart-cost-edge} for an example). Now consider computing  $\min_{(e'_i, e_j)} \cut(e'_i, e_j)$. To do this, define a matrix $M$ whose entry at the $i^{th}$ row and $j^{th}$ column is $M[i,j]=\cut(e'_i, e_j)$ (see the matrix in \Cref{fig:bipart-cost-edge} for an example).  Computing  $\min_{(e'_i, e_j)} \cut(e'_i, e_j)$ becomes computing $\min_{i,j} M[i,j]$. The key to solve this problem is the following property.

\begin{figure}[h]
	\centering
	\includegraphics[scale=0.8]{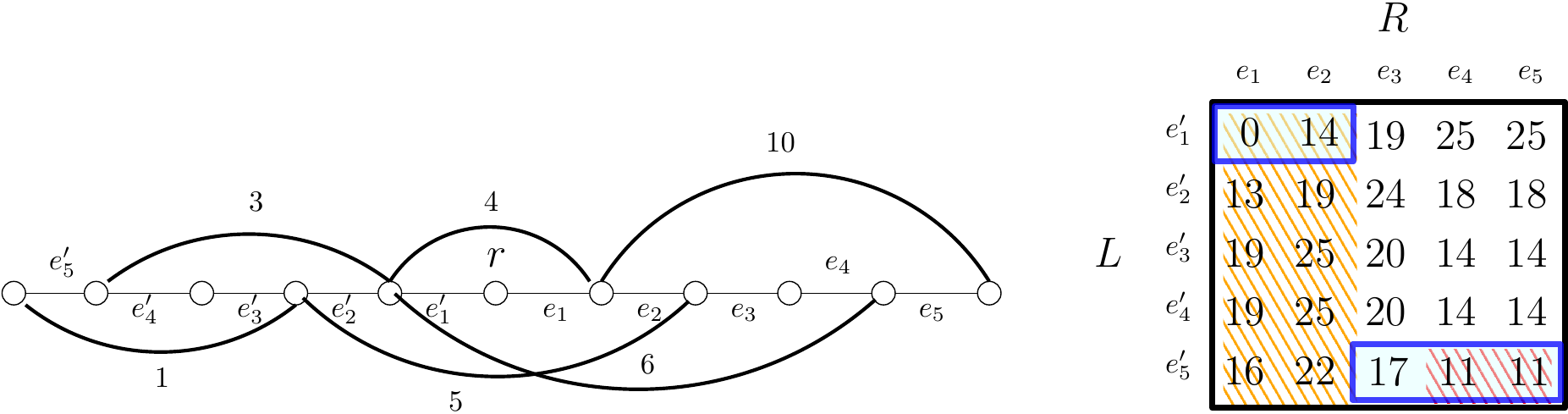}
	\caption{Example of when $T$ is a path. In the graph, $e_1, \ldots, e_5, e'_1, \ldots, e'_5$ are edges in $T$ and have weight $0$. Numbers on other edges indicate the edge weights.}
	\label{fig:bipart-cost-edge}
\end{figure}

\begin{observation}\label{thm:intro:pathStructure}
Assume for simplicity that the minimum entry in each column of $M$ is unique. For any $j$, let $r^*(j)=\arg\min_i M[i, j]$. Then, $r^*(j)\geq r^*(j-1)$ for any $j$. That is, $r^*$ is a non-decreasing function. 
\end{observation}

For example, the minimum entries of each column of the matrix in \Cref{fig:bipart-cost-edge} are in solid (blue) rectangles. Observe that the positions of these entries do not move up when we scan the columns from left to right. 

We postpone proving Observation \ref{thm:intro:pathStructure} to later. We now show how to exploit it to develop a divide-and-conquer algorithm. Our algorithm probes all entries in the $\ell$-th column, where $\ell$ denote the index of the middle column of $M$ ($\ell=3$ in \Cref{fig:bipart-cost-edge}). It can then compute $r^*(\ell)$. By Observation \ref{thm:intro:pathStructure}, it suffices to recurse the problem on the sub-matrices $M_1=M[1, \ldots, r^*(\ell); 1, \ldots, \ell-1]$ and $M_2=M[r^*(\ell), \ldots ; \ell+1, \ldots]$; i.e., entries of $M_1$ are those in rows $i\leq r^*(\ell)$ and columns $j<\ell$ and entries of $M_2$ are those in rows $i\geq r^*(\ell)$ and columns $j>\ell$ (for example, the shaded (orange) areas in \Cref{fig:bipart-cost-edge}). It is not hard to see that this algorithm only probes $\tO(n)$ values in $M$.   

The above algorithm can be easily extended to the case where we do not have $r$: Pick a middle node of the path $T$ as $r$ and run the algorithm above. Then, make a recursion on all edges on the left of $r$ and another on all edges on the right side of $r$. It is not hard to see that this algorithm requires $\tO(n)$ probes. 

\paragraph{When $T$ is a star-graph.} Now we consider another extreme case, where $T$ is a star-graph; i.e., the spanning tree $T$ has $n-1$ many disjoint edges of the form $e_i = (r, i)$ where $r$ is the root of the tree. 
In this case, we exploit the fact that a cut that contains only one edge of $T$ is a candidate for the 2-respecting min-cut too. Let  $\deg(u)$ denote the degree of node $u$ and $C(u,v)$ denote the weight of edge $(u,v)$.

\begin{observation}\label{thm:intro:starStructure}
Assume that no 2-respecting min-cut contains only one tree-edge.
Let $(e_i, e_j)=\arg\min_{(e_{i'}, e_{j'})}$ $\cut(e_{i'}, e_{j'})$; i.e. the 2-respecting min-cuts have weight $\cut(e_i, e_j)$.
Then, $d(i)<2C(i,j)$ and  $d(j)<2C(i,j)$. 
\end{observation}
\begin{proof}
	Note that $\cut(e_i, e_j)=\deg(i)+\deg(j)-2C(i,j)$. 
	Consider the cut that separates node $i$ from other nodes. This cut has weight $\deg(i)$. Since this is not a 2-respecting min-cut (by the assumption), $\deg(i)> \cut(e_i, e_j)=\deg(i)+\deg(j)-2C(i,j)$; thus, $d(j)<2C(i,j)$. Similarly, $d(i)<2C(i,j)$. 
\end{proof}



Observation \ref{thm:intro:starStructure} motivates the following notion. We say that a tree-edge $e_i$ is {\em interested in} a tree-edge $e_j$ if $\deg(i)<2C(i,j)$. Note that each tree-edge $e_i$ can be interested in at most one tree-edge: if $e_i$ is interested in two distinct edges $e_j$ and $e_{j}'$, then the degree of $e_i$ is at least $C(i,j)+C(i,j')>(\deg(i)/2)+(\deg(i)/2)$, a contradiction. It is not hard to efficiently find the tree-edge that $e_i$ is interested in, for every tree-edge $e_i$. (We only need to consider the maximum-weight edge incident to each node.) Our algorithm is now the following. For each edge $e_i$, let $e_j$ be the edge that it is interested in. Compute $\cut(e_i,e_j)$. Output the minimum value among (I) the values of all the computed 2-respecting cuts and (II) the minimum degree.
By Observation \ref{thm:intro:starStructure}, if (II)  does not give the 2-respecting min-cut value, then (I) does. 
See \Cref{fig:star-interesting} for an example.

\begin{figure}[h]
	\centering
	\includegraphics[scale=0.5]{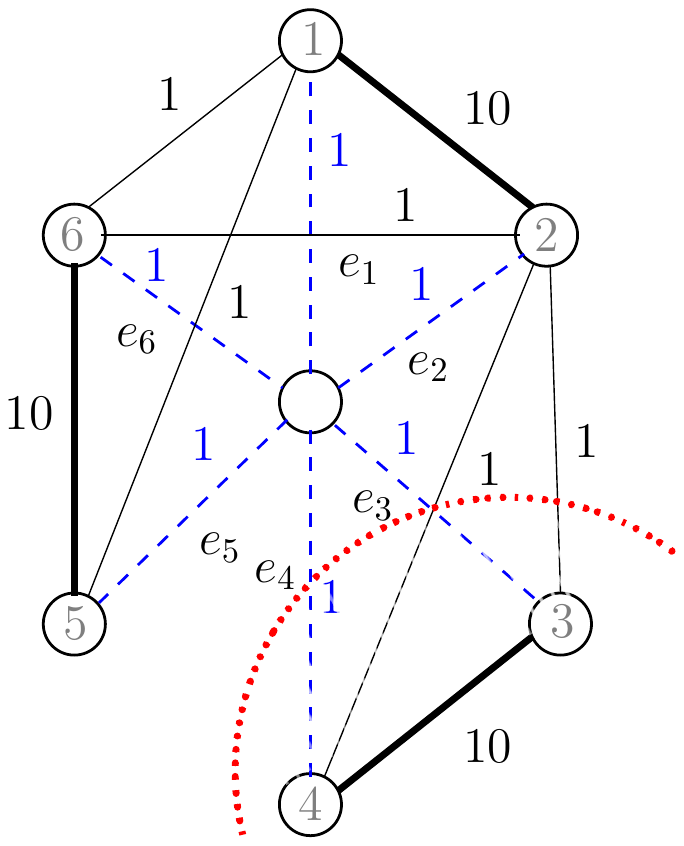}
	\caption{An example where $T$ is a star-graph. Dashed edges (in blue) are tree edges. Numbers on edges indicate their weights. (All but bold edges have weight 1.) The dashed curve (in red) indicates the min-cut of the graph. This min-cut contains two tree-edges, namely $e_3$ and $e_4$. Observe that the following pairs of tree-edges are interested in each other $(e_1, e_2)$, $(e_3,e_4)$ and $(e_5,e_6)$. Our algorithm outputs the minimum among $\min_i(\deg(i))$, $\cut(e_1, e_2)$, $\cut(e_3,e_4)$ and $\cut(e_5,e_6)$.}
	\label{fig:star-interesting}
\end{figure}

\paragraph{Handling general tree $T$.} Since it is easy to find a 2-respecting min-cut if it  contains only one edge in $T$, we assume from now that all 2-respecting min-cut contains exactly two edges in $T$. Let $(e^*, f^*)=\arg\min_{(e, f)} \cut(e, f)$; i.e. $e^*$ and $f^*$ are tree-edges such that the 2-respecting min-cuts have weight $\cut(e^*, f^*)$. For simplicity, let us assume further that $e^*$ and $f^*$ are {\em orthogonal} in that $e^*$ does not lie in the unique path in $T$ between $f^*$ and the root of $T$, and vice versa. 

The ideas behind Observation \ref{thm:intro:starStructure} can be naturally extended to a general tree $T$ as follows. For any node $u$, let $u^\da$ denote the set of nodes in the sub-tree of $T$ rooted at $u$. Let $\deg(u^\da)$ be the total weight of edges between $u^\da$ and $V(G)\setminus u^\da$. For any nodes $u$ and $v$, let $C(u^\da, v^\da)$ be the total weight of edges between $u^\da$ and $v^\da$. Consider any tree-edges $(x, y)$ and $(x', y')$, where $x$ (and $x'$) is the parent of $y$ (and $y'$ respectively) in $T$. We say that $(x,y)$ is {\em interested in} $(x',y')$ if  
\begin{align}\label{eq:intro:interestingEdgeDef}
\deg(y^\da)<2C(y^\da,y'^\da).
\end{align}

It is not hard to show that the tree-edges $e^*$ and $f^*$ defined above are interested in each other. Thus, like in the case of star-graphs, we can find the 2-respecting min-cut by computing $\cut(e,f)$ for all pairs of $e$ and $f$ that are interested in each other.  
However, this does not imply that we need to compute only $\tilde O(n)$ many values of $\cut(e,f)$, since an edge can be interested in {\em many} other edges. An additional helpful property is this:

\begin{observation}\label{thm:intro:interestingPath}
For any tree-edge $e=(x,y)$, edges that $e$ is interested in form a path in $T$ between the root and some node $v$. 
\end{observation}
\begin{proof}
Throughout the proof, let $e'=(x',y')$ and $e''=(x'',y'')$ be tree-edges where $x'$ (and $x''$) is the parent of $y'$  (and $y''$ respectively) in $T$. 
If $e'$ lies in the path between the root and $e''$ and $e$ is interested in $e''$, then $e$ is also interested in any $e'$ because $C(y^\da,y'^\da)\geq C(y^\da,y''^\da)$. 
%
%
%
Now we show that if $e'$ and $e''$ are orthogonal (i.e. $y'\notin y''^\da$ and vice versa), then $e$ cannot be interested in both of them.
Since $e'$ and $e''$ are orthogonal, the set of edges between $y^\da$ and $y'^\da$ is disjoint from the set of edges between $y^\da$ and $y''^\da$. So, $\deg(y^\da)\ge C(y, y'^\da)+C(y, y''^\da)$. Now, if  $e$ is interested in both $e'$ and $e''$, we have $C(y, y'^\da)+C(y, y''^\da)>\deg(y^\da)$, a contradiction. 
\end{proof}


\danupon{Below was changed after discussing with Bryce Sandlund}

Our last idea is to exploit the above property by decomposing a tree into paths where each root-to-leaf path contains $O(log n)$ decomposed paths. A tree decomposition of this type was used in the sequential setting in Karger's algorithm. 
Later Lovett and Sandlund used the well-known heavy-light decomposition to simplify Karger's algorithm. (The heavy-light decomposition was also internally used in some data structures used by Karger.) 
Both decompositions work equally well for our algorithms. To the best of our knowledge, decompositions of this type have not been used before in query complexity and streaming algorithms.


Roughly, the heavy-light decomposition partitions the tree-edges into a family $\cP$ of paths, such that every root-to-leaf path $P$ in the tree $T$ shares edges with at most $O(\log n)$ many paths in $\cP$. Thus, Observation \ref{thm:intro:interestingPath} implies that a tree-edge $e$ is interested in edges lying in $O(\log n)$ many paths in $\cP$. 
This motivates the following algorithm: We say that a tree-edge $e$ is {\em interested in a path $P$} in $\cP$ if it is interested in some edge in $P$. For any two distinct paths $P$ and $Q$ in $\cP$, do the following. Imagine that we contract all-tree edges except those in $P$ that are interested in $Q$ and those in $Q$ that are interested in $P$. Suppose that we are left with a path $P'$. We execute the path algorithm on $P'$. (Note that in reality we might be left with the case where every node has degree at most two except one node whose degree is three. This can be handled similarly. Also, when we implement this algorithm in different settings we do not actually have to contract edges. We only have to simulate the path algorithm on the edges in $P'$.)


Recall that  when we run the path algorithm on each path $P'$ above we need to compute $\cut(e,f)$ for $\tilde O(|E(P')|)$ many pairs of tree-edges $(e,f)$, where $|E(P')|$ is the number of edges in $P'$.  
%
Moreover, each tree-edge participates in $O(\log n)$ such paths since it is interested in $O(\log n)$ many paths in $\cP$. Thus, the total  number of cut values $\cut(e,f)$ that we need to compute for the above algorithm is $\tilde O(n)$ in total. To see why the algorithm finds a 2-respecting min-cut, consider when $P$ and $Q$ contains $e^*$ and $f^*$ respectively. After the contractions $e^*$ and $f^*$ remain in $P'$ and thus the path algorithm executed on $P'$ finds $\cut(e^*,f^*)$.

\paragraph{Finding paths that an edge is interested in.} Most steps described above can be implemented quite easily in all settings. The step that is sometimes tricky is finding paths in the heavy-light decomposition $\cP$ that a tree edge $e$ is interested in. In the cut-query and streaming settings, we can compute this from a {\em cut sparsifier} which can be computed efficiently  (e.g. \cite{AhnGM12-pods,AhnGM12-soda,GoelKP12,RubinsteinSW18}). Since a cut sparsifier preserves all cuts approximately and since we are looking for $(y, y')$ such that $C(y^\da, y'^\da)$ is large compared to $\deg(y^\da)$ (as in \Cref{eq:intro:interestingEdgeDef}), we can use the cut sparsifier to identify $(y, y')$ that ``might'' satisfy \Cref{eq:intro:interestingEdgeDef}. This means that for each edge $e$, we can identify a set of ``potential'' edges such that some of these edges will actually interest $e$. These edges might not form a path as in Observation \ref{thm:intro:interestingPath}, but we can show that they form only $O(1)$ paths.


For our sequential algorithm, let us assume for simplicity that the input graph is unweighted (possibly with parallel edges). For any node $y$ we find $y'$ that {\em might} satisfy  \Cref{eq:intro:interestingEdgeDef} by sampling $\Theta(\log n)$ edges among edges between $y^\da$ and $V(G)\setminus y^\da$. For $y'$ that satisfies \Cref{eq:intro:interestingEdgeDef} there is a sampled edge between $y^\da$ and $y'^\da$ with high probability because $C(y^\da,y'^\da)>\deg(y^\da)/2$. Thus, it suffices to check, for every $O(\log n)$ sampled edge $e'$ and every path $P$ in $\cP$ that overlaps with the tree-path from the root to $e$, whether $e$ is interested in $P$ or not. (Checking this requires computing $\cut(u,v)$ for $O(1)$ many pairs of nodes $(a,b)$.) To sample the edges, we use the {\em range sampling} data structure in a way similar to how we use the range counting data structure to compute the cut size as outlined above. We build a range sampling data structure using the {\em range reporting} data structure by Overmars and Chazelle \cite{Overmars88,Chaz88}. We need $O(m\log m)$ pre-processing time and $O(\log n)$ amortized time to report or sample each point (edge). (With bit operations, the preprocessing time becomes $O(m\sqrt{\log m})$ using \cite{BGKS15,CP10,MunroNV16}.) 





%% file: main.bbl
\newcommand{\etalchar}[1]{$^{#1}$}
\begin{thebibliography}{AKM{\etalchar{+}}87}

\bibitem[ACK19]{AssadiCK19}
Sepehr Assadi, Yu~Chen, and Sanjeev Khanna.
\newblock Polynomial pass lower bounds for graph streaming algorithms.
\newblock In {\em {STOC}}, pages 265--276. {ACM}, 2019.

\bibitem[AGM12a]{AhnGM12-soda}
Kook~Jin Ahn, Sudipto Guha, and Andrew McGregor.
\newblock Analyzing graph structure via linear measurements.
\newblock In {\em {SODA}}, pages 459--467. {SIAM}, 2012.

\bibitem[AGM12b]{AhnGM12-pods}
Kook~Jin Ahn, Sudipto Guha, and Andrew McGregor.
\newblock Graph sketches: sparsification, spanners, and subgraphs.
\newblock In {\em {PODS}}, pages 5--14. {ACM}, 2012.

\bibitem[AKM{\etalchar{+}}87]{AggarwalKMSW87}
Alok Aggarwal, Maria~M. Klawe, Shlomo Moran, Peter~W. Shor, and Robert~E.
  Wilber.
\newblock Geometric applications of a matrix-searching algorithm.
\newblock {\em Algorithmica}, 2:195--208, 1987.

\bibitem[BGKS15]{BGKS15}
Maxim~A. Babenko, Pawel Gawrychowski, Tomasz Kociumaka, and Tatiana
  Starikovskaya.
\newblock Wavelet trees meet suffix trees.
\newblock In {\em Proceedings of the 26th SODA}, pages 572--591, 2015.

\bibitem[BKR96]{BurkardKR96}
Rainer~E. Burkard, Bettina Klinz, and R{\"{u}}diger Rudolf.
\newblock Perspectives of monge properties in optimization.
\newblock {\em Discret. Appl. Math.}, 70(2):95--161, 1996.

\bibitem[CGL{\etalchar{+}}19]{ChuzhoyGLNPS19}
Julia Chuzhoy, Yu~Gao, Jason Li, Danupon Nanongkai, Richard Peng, and
  Thatchaphol Saranurak.
\newblock A deterministic algorithm for balanced cut with applications to
  dynamic connectivity, flows, and beyond.
\newblock {\em CoRR}, abs/1910.08025, 2019.

\bibitem[Cha88]{Chaz88}
Bernard Chazelle.
\newblock A functional approach to data structures and its use in
  multidimensional searching.
\newblock {\em {SIAM} J. Comput.}, 17(3):427--462, 1988.

\bibitem[CP10]{CP10}
Timothy~M. Chan and Mihai Patrascu.
\newblock Counting inversions, offline orthogonal range counting, and related
  problems.
\newblock In {\em Proceedings of the 21st {SODA}}, pages 161--173, 2010.

\bibitem[DEMN21]{DoryEMN21}
Michal Dory, Yuval Efron, Sagnik Mukhopadhyay, and Danupon Nanongkai.
\newblock Distributed weighted min-cut in nearly-optimal time.
\newblock {\em STOC}, 2021.

\bibitem[FKM{\etalchar{+}}05]{FeigenbaumKMSZ05}
Joan Feigenbaum, Sampath Kannan, Andrew McGregor, Siddharth Suri, and Jian
  Zhang.
\newblock On graph problems in a semi-streaming model.
\newblock {\em Theor. Comput. Sci.}, 348(2-3):207--216, 2005.

\bibitem[GKP12]{GoelKP12}
Ashish Goel, Michael Kapralov, and Ian Post.
\newblock Single pass sparsification in the streaming model with edge
  deletions.
\newblock {\em CoRR}, abs/1203.4900, 2012.

\bibitem[GMW20]{GMW19}
Pawel Gawrychowski, Shay Mozes, and Oren Weimann.
\newblock Minimum cut in o(m log{\({^2}\)} n) time.
\newblock In {\em {ICALP}}, volume 168 of {\em LIPIcs}, pages 57:1--57:15.
  Schloss Dagstuhl - Leibniz-Zentrum f{\"{u}}r Informatik, 2020.

\bibitem[GMW21]{gawrychowski2021note}
Pawe{\l} Gawrychowski, Shay Mozes, and Oren Weimann.
\newblock A note on a recent algorithm for minimum cut.
\newblock In {\em Symposium on Simplicity in Algorithms (SOSA)}, pages 74--79.
  SIAM, 2021.

\bibitem[GNT20]{GhaffariNT20}
Mohsen Ghaffari, Krzysztof Nowicki, and Mikkel Thorup.
\newblock Faster algorithms for edge connectivity via random 2-out
  contractions.
\newblock In {\em {SODA}}. {SIAM}, 2020.

\bibitem[HdLT01]{HolmLT01}
Jacob Holm, Kristian de~Lichtenberg, and Mikkel Thorup.
\newblock Poly-logarithmic deterministic fully-dynamic algorithms for
  connectivity, minimum spanning tree, 2-edge, and biconnectivity.
\newblock {\em J. {ACM}}, 48(4):723--760, 2001.
\newblock Announced at STOC 1998.

\bibitem[HK99]{HenzingerK99}
Monika~Rauch Henzinger and Valerie King.
\newblock Randomized fully dynamic graph algorithms with polylogarithmic time
  per operation.
\newblock {\em J. {ACM}}, 46(4):502--516, 1999.
\newblock Announced at STOC 1995.

\bibitem[HRW17]{HenzingerRW17}
Monika Henzinger, Satish Rao, and Di~Wang.
\newblock Local flow partitioning for faster edge connectivity.
\newblock In {\em {SODA}}, pages 1919--1938. {SIAM}, 2017.

\bibitem[HT84]{HT84}
Dov Harel and Robert~Endre Tarjan.
\newblock Fast algorithms for finding nearest common ancestors.
\newblock {\em {SIAM} J. Comput.}, 13(2):338--355, 1984.

\bibitem[Kar95]{Karger_thesis}
David~Ron Karger.
\newblock {\em Random Sampling in Graph Optimization Problems}.
\newblock PhD thesis, Stanford, CA, USA, 1995.
\newblock UMI Order No. GAX95-16851.

\bibitem[Kar99]{Karger99-skeleton}
David~R. Karger.
\newblock Random sampling in cut, flow, and network design problems.
\newblock {\em Math. Oper. Res.}, 24(2):383--413, 1999.
\newblock announced at STOC'94.

\bibitem[Kar00]{Kar00}
David~R. Karger.
\newblock Minimum cuts in near-linear time.
\newblock {\em J. {ACM}}, 47(1):46--76, 2000.
\newblock announced at STOC'96.

\bibitem[KKM13]{KapronKM13}
Bruce~M. Kapron, Valerie King, and Ben Mountjoy.
\newblock Dynamic graph connectivity in polylogarithmic worst case time.
\newblock In {\em Proceedings of the Twenty-Fourth Annual {ACM-SIAM} Symposium
  on Discrete Algorithms, {SODA} 2013, New Orleans, Louisiana, USA, January
  6-8, 2013}, pages 1131--1142, 2013.

\bibitem[KLM{\etalchar{+}}17]{KLMMS17}
Michael Kapralov, Yin~Tat Lee, Cameron Musco, Christopher Musco, and Aaron
  Sidford.
\newblock Single pass spectral sparsification in dynamic streams.
\newblock {\em {SIAM} J. Comput.}, 46(1):456--477, 2017.

\bibitem[KT19]{KawarabayashiT19}
Ken{-}ichi Kawarabayashi and Mikkel Thorup.
\newblock Deterministic edge connectivity in near-linear time.
\newblock {\em J. {ACM}}, 66(1):4:1--4:50, 2019.
\newblock announced at STOC'15.

\bibitem[LMN21]{lopezmn21}
Andrés {López-Martínez}, Sagnik Mukhopadhyay, and Danupon Nanongkai.
\newblock Work-optimal parallel minimum cuts for non-sparse graphs, 2021.

\bibitem[LS19]{LovettS19}
Antonio~Molina Lovett and Bryce Sandlund.
\newblock A simple algorithm for minimum cuts in near-linear time.
\newblock {\em CoRR}, abs/1908.11829, 2019.

\bibitem[Mat93]{Matula93}
David~W. Matula.
\newblock A linear time 2+epsilon approximation algorithm for edge
  connectivity.
\newblock In {\em {SODA}}, pages 500--504. {ACM/SIAM}, 1993.

\bibitem[MNV16]{MunroNV16}
J.~Ian Munro, Yakov Nekrich, and Jeffrey~Scott Vitter.
\newblock Fast construction of wavelet trees.
\newblock {\em Theor. Comput. Sci.}, 638:91--97, 2016.

\bibitem[NSW17]{NanongkaiSW17}
Danupon Nanongkai, Thatchaphol Saranurak, and Christian Wulff{-}Nilsen.
\newblock Dynamic minimum spanning forest with subpolynomial worst-case update
  time.
\newblock In {\em {FOCS}}, pages 950--961. {IEEE} Computer Society, 2017.

\bibitem[Ove88]{Overmars88}
Mark~H. Overmars.
\newblock Efficient data structures for range searching on a grid.
\newblock {\em J. Algorithms}, 9(2):254--275, 1988.

\bibitem[RSW18]{RubinsteinSW18}
Aviad Rubinstein, Tselil Schramm, and S.~Matthew Weinberg.
\newblock Computing exact minimum cuts without knowing the graph.
\newblock In {\em Proceedings of the 9th {ITCS}}, pages 39:1--39:16, 2018.

\bibitem[ST83]{ST83}
Daniel~Dominic Sleator and Robert~Endre Tarjan.
\newblock A data structure for dynamic trees.
\newblock {\em J. Comput. Syst. Sci.}, 26(3):362--391, 1983.

\bibitem[Tho07]{Thorup07}
Mikkel Thorup.
\newblock Fully-dynamic min-cut.
\newblock {\em Combinatorica}, 27(1):91--127, 2007.

\bibitem[Vit85]{Vit85}
Jeffrey~Scott Vitter.
\newblock Random sampling with a reservoir.
\newblock {\em ACM Transactions on Mathematical Software}, 11:37--57, 1985.

\bibitem[Wul13]{Wulff-Nilsen13a}
Christian Wulff{-}Nilsen.
\newblock Faster deterministic fully-dynamic graph connectivity.
\newblock In {\em Proceedings of the Twenty-Fourth Annual {ACM-SIAM} Symposium
  on Discrete Algorithms, {SODA} 2013, New Orleans, Louisiana, USA, January
  6-8, 2013}, pages 1757--1769, 2013.

\bibitem[Zel11]{Zelke11}
Mariano Zelke.
\newblock Intractability of min- and max-cut in streaming graphs.
\newblock {\em Inf. Process. Lett.}, 111(3):145--150, 2011.

\end{thebibliography}
